\documentclass[letterpaper,twocolumn,10pt]{article}
\usepackage{usenix-2020-09}

\usepackage{array}
\usepackage{xcolor}
\usepackage{amsthm}
\usepackage{algorithm}
\usepackage{algorithmic}
\usepackage{booktabs}
\usepackage{multirow}
\usepackage{subcaption}
\usepackage{pifont}
\usepackage{enumitem}
\usepackage{amsmath,amssymb,amsfonts}
\usepackage{graphicx}
\usepackage{textcomp}
\usepackage{cite}
\usepackage{url}
\usepackage{hyperref}

\definecolor{rqblue}{RGB}{31,78,121}
\DeclareRobustCommand{\RQ}[1]{{\setlength{\fboxsep}{1.2pt}\colorbox{rqblue}{\textcolor{white}{\textsf{\textbf{RQ#1}}}}}}

\setitemize{noitemsep,topsep=1.5pt,parsep=0pt,partopsep=0pt, leftmargin=*}
\setenumerate{noitemsep,topsep=1.5pt,parsep=0pt,partopsep=0pt, leftmargin=*}

\newif\ifcomment
\commenttrue

\ifcomment
    \newcounter{YXNumberOfComments}
    \stepcounter{YXNumberOfComments}
    \newcommand{\xym}[1]{\textcolor{orange}{\small \bf [XYM\#\arabic{YXNumberOfComments}\stepcounter{YXNumberOfComments}: #1]}}

    \newcounter{PHNumberOfComments}
    \stepcounter{PHNumberOfComments}
    \newcommand{\hpc}[1]{\textcolor{blue}{\small \bf [HPC\#\arabic{PHNumberOfComments}\stepcounter{PHNumberOfComments}: #1]}}  
    
    \newcounter{LMYNumberOfComments}
    \stepcounter{LMYNumberOfComments}
    \newcommand{\lmy}[1]{\textcolor{red}{\small \bf [LMY\#\arabic{LMYNumberOfComments}\stepcounter{LMYNumberOfComments}: #1]}} 
    
    \newcommand{\del}[1]{{{\color{red}\st{#1}}}}
\else
    \newcommand\xym[1]{}    
    \newcommand\hpc[1]{} 
    \newcommand\lmy[1]{} 
    \newcommand{\del}[1]{}
\fi

\newtheorem{definition}{Definition}
\newtheorem{theorem}{Theorem}
\newtheorem{lemma}[theorem]{Lemma}

\newcommand{\bcode}{\mathbf{b}}

\newcommand{\yesmark}{\ding{51}}
\newcommand{\dpmark}{\textbf{\textsc{dp}}}
\newcommand{\partialmark}{$\triangle$}
\newcommand{\nomark}{\ding{55}}

\DeclareMathOperator{\sign}{sign}

\newcommand{\DPMainUtilityTable}{%
\begin{table*}[t]
\centering
\caption{Representative two-stage DP operating points. NDCG@10 averages E5 and BGE on SciDocs, NQ, and FEVER. RAG latency uses Qwen3-32B (vLLM) with 128 output tokens, median over 100 NQ queries; $\Delta$ is the absolute protection cost over the plaintext two-stage pipeline.}
\label{tab:dp-retrieval}
\footnotesize
\resizebox{0.7\linewidth}{!}{%
\begin{tabular}{@{}llrrrrr@{}}
\toprule
Release & Guarantee & $(\varepsilon,K)$ & NDCG@10 & Ret. & RAG s & $\Delta$s \\
\midrule
Full-corpus float & Reference & $(\infty,N)$ & .5394 & 100.0\% & 5.750 & --- \\
Plaintext filter & None & $(\infty,500)$ & .5363 & 99.4\% & 5.753 & 0 \\
\midrule
Randomized response & $(\varepsilon,\delta)$-DP & $(16,3000)$ & .0306 & 5.7\% & 6.319 & +0.566 \\
Gaussian & $(\varepsilon,\delta)$-DP & $(8,3000)$ & .5084 & 94.3\% & 6.319 & +0.566 \\
Gaussian & $(\varepsilon,\delta)$-DP & $(16,2000)$ & .5350 & 99.2\% & 6.128 & +0.375 \\
RDP-vMF & $(\varepsilon,\delta)$-DP & $(8,3000)$ & .5187 & 96.2\% & 6.319 & +0.566 \\
RDP-vMF & $(\varepsilon,\delta)$-DP & $(16,1000)$ & .5350 & 99.2\% & 5.959 & +0.206 \\
Pure-vMF & Pure metric DP & $(32,3000)$ & .5229 & 96.9\% & 6.319 & +0.566 \\
Pure-vMF & Pure metric DP & $(64,2000)$ & .5359 & 99.4\% & 6.128 & +0.375 \\
\bottomrule
\end{tabular}
}
\end{table*}%
}

\newcommand{\DPEFiveFullTable}{%
\begin{table*}[t]
\centering
\caption{Complete E5 two-stage NDCG@10 under query-side DP. Randomized response, Gaussian, and RDP-vMF use $\delta=10^{-6}$; pure-vMF gives the formal pure metric-DP guarantee.}
\label{tab:dp-e5-full}
\small
\renewcommand{\arraystretch}{1.0}
\setlength{\tabcolsep}{5pt}
\begin{tabular}{@{}lllrrrrrr@{}}
\toprule
Dataset & Mechanism & $\varepsilon$ & $K=200$ & $K=500$ & $K=1000$ & $K=2000$ & $K=3000$ & Float \\
\midrule
\multirow{17}{*}{SciDocs} & No DP & $\infty$ & .1884 & .1881 & .1877 & .1872 & .1870 & .1870 \\
 & Randomized response & 8 & .0102 & .0203 & .0308 & .0474 & .0596 & .1870 \\
 & Randomized response & 16 & .0111 & .0235 & .0380 & .0536 & .0700 & .1870 \\
 & Randomized response & 32 & .0250 & .0458 & .0632 & .0870 & .1057 & .1870 \\
 & Randomized response & 64 & .0566 & .0766 & .0992 & .1262 & .1407 & .1870 \\
 & Gaussian & 8 & .1558 & .1685 & .1784 & .1829 & .1844 & .1870 \\
 & Gaussian & 16 & .1857 & .1868 & .1873 & .1880 & .1876 & .1870 \\
 & Gaussian & 32 & .1853 & .1883 & .1875 & .1874 & .1872 & .1870 \\
 & Gaussian & 64 & .1885 & .1871 & .1875 & .1872 & .1873 & .1870 \\
 & RDP-vMF & 8 & .1630 & .1776 & .1820 & .1834 & .1854 & .1870 \\
 & RDP-vMF & 16 & .1833 & .1850 & .1872 & .1869 & .1866 & .1870 \\
 & RDP-vMF & 32 & .1860 & .1871 & .1877 & .1874 & .1872 & .1870 \\
 & RDP-vMF & 64 & .1874 & .1877 & .1880 & .1874 & .1870 & .1870 \\
 & Pure-vMF & 8 & .0486 & .0750 & .0944 & .1199 & .1363 & .1870 \\
 & Pure-vMF & 16 & .1210 & .1447 & .1599 & .1732 & .1804 & .1870 \\
 & Pure-vMF & 32 & .1757 & .1851 & .1848 & .1863 & .1866 & .1870 \\
 & Pure-vMF & 64 & .1846 & .1854 & .1857 & .1869 & .1868 & .1870 \\
\midrule
\multirow{17}{*}{NQ} & No DP & $\infty$ & .5694 & .5770 & .5801 & .5820 & .5826 & .5854 \\
 & Randomized response & 8 & .0003 & .0005 & .0008 & .0015 & .0043 & .5854 \\
 & Randomized response & 16 & .0018 & .0025 & .0041 & .0069 & .0096 & .5854 \\
 & Randomized response & 32 & .0044 & .0094 & .0155 & .0264 & .0326 & .5854 \\
 & Randomized response & 64 & .0379 & .0619 & .0873 & .1170 & .1401 & .5854 \\
 & Gaussian & 8 & .4440 & .4845 & .5072 & .5271 & .5389 & .5854 \\
 & Gaussian & 16 & .5519 & .5657 & .5708 & .5761 & .5778 & .5854 \\
 & Gaussian & 32 & .5663 & .5750 & .5783 & .5811 & .5820 & .5854 \\
 & Gaussian & 64 & .5689 & .5781 & .5800 & .5821 & .5825 & .5854 \\
 & RDP-vMF & 8 & .4645 & .5015 & .5257 & .5434 & .5522 & .5854 \\
 & RDP-vMF & 16 & .5589 & .5695 & .5752 & .5765 & .5778 & .5854 \\
 & RDP-vMF & 32 & .5681 & .5761 & .5796 & .5822 & .5827 & .5854 \\
 & RDP-vMF & 64 & .5701 & .5764 & .5800 & .5821 & .5832 & .5854 \\
 & Pure-vMF & 8 & .0254 & .0414 & .0592 & .0852 & .1043 & .5854 \\
 & Pure-vMF & 16 & .2353 & .2971 & .3443 & .3985 & .4219 & .5854 \\
 & Pure-vMF & 32 & .5061 & .5328 & .5465 & .5571 & .5617 & .5854 \\
 & Pure-vMF & 64 & .5596 & .5685 & .5732 & .5781 & .5801 & .5854 \\
\midrule
\multirow{17}{*}{FEVER} & No DP & $\infty$ & .8392 & .8441 & .8458 & .8472 & .8475 & .8501 \\
 & Randomized response & 8 & .0007 & .0008 & .0013 & .0022 & .0024 & .8501 \\
 & Randomized response & 16 & .0000 & .0005 & .0011 & .0024 & .0039 & .8501 \\
 & Randomized response & 32 & .0039 & .0060 & .0103 & .0162 & .0226 & .8501 \\
 & Randomized response & 64 & .0343 & .0542 & .0762 & .1037 & .1264 & .8501 \\
 & Gaussian & 8 & .6029 & .6675 & .7083 & .7410 & .7561 & .8501 \\
 & Gaussian & 16 & .8161 & .8255 & .8319 & .8380 & .8412 & .8501 \\
 & Gaussian & 32 & .8389 & .8436 & .8465 & .8478 & .8481 & .8501 \\
 & Gaussian & 64 & .8396 & .8435 & .8470 & .8480 & .8482 & .8501 \\
 & RDP-vMF & 8 & .6484 & .7037 & .7406 & .7690 & .7867 & .8501 \\
 & RDP-vMF & 16 & .8246 & .8332 & .8390 & .8428 & .8452 & .8501 \\
 & RDP-vMF & 32 & .8385 & .8427 & .8460 & .8485 & .8488 & .8501 \\
 & RDP-vMF & 64 & .8396 & .8439 & .8464 & .8476 & .8483 & .8501 \\
 & Pure-vMF & 8 & .0187 & .0331 & .0503 & .0713 & .0890 & .8501 \\
 & Pure-vMF & 16 & .2753 & .3489 & .4098 & .4722 & .5097 & .8501 \\
 & Pure-vMF & 32 & .7293 & .7665 & .7908 & .8087 & .8171 & .8501 \\
 & Pure-vMF & 64 & .8260 & .8350 & .8400 & .8434 & .8447 & .8501 \\
\bottomrule
\end{tabular}
\end{table*}%
}

\newcommand{\DPBGEFullTable}{%
\begin{table*}[t]
\centering
\caption{Complete BGE two-stage NDCG@10 under query-side DP. Randomized response, Gaussian, and RDP-vMF use $\delta=10^{-6}$; Pure-vMF gives the formal pure metric-DP guarantee.}
\label{tab:dp-bge-full}
\small
\renewcommand{\arraystretch}{1.0}
\setlength{\tabcolsep}{5pt}
\begin{tabular}{@{}lllrrrrrr@{}}
\toprule
Dataset & Mechanism & $\varepsilon$ & $K=200$ & $K=500$ & $K=1000$ & $K=2000$ & $K=3000$ & Float \\
\midrule
\multirow{17}{*}{SciDocs} & No DP & $\infty$ & .2233 & .2225 & .2227 & .2229 & .2228 & .2228 \\
 & Randomized response & 8 & .0098 & .0222 & .0344 & .0552 & .0752 & .2228 \\
 & Randomized response & 16 & .0169 & .0335 & .0486 & .0672 & .0861 & .2228 \\
 & Randomized response & 32 & .0266 & .0473 & .0671 & .0953 & .1111 & .2228 \\
 & Randomized response & 64 & .0635 & .0984 & .1295 & .1534 & .1731 & .2228 \\
 & Gaussian & 8 & .2073 & .2180 & .2202 & .2228 & .2230 & .2228 \\
 & Gaussian & 16 & .2221 & .2231 & .2229 & .2230 & .2229 & .2228 \\
 & Gaussian & 32 & .2227 & .2226 & .2223 & .2228 & .2227 & .2228 \\
 & Gaussian & 64 & .2227 & .2224 & .2230 & .2229 & .2228 & .2228 \\
 & RDP-vMF & 8 & .2174 & .2220 & .2218 & .2230 & .2233 & .2228 \\
 & RDP-vMF & 16 & .2222 & .2227 & .2232 & .2231 & .2229 & .2228 \\
 & RDP-vMF & 32 & .2222 & .2224 & .2230 & .2229 & .2228 & .2228 \\
 & RDP-vMF & 64 & .2228 & .2227 & .2230 & .2227 & .2228 & .2228 \\
 & Pure-vMF & 8 & .0580 & .0822 & .1094 & .1434 & .1651 & .2228 \\
 & Pure-vMF & 16 & .1530 & .1784 & .1947 & .2097 & .2169 & .2228 \\
 & Pure-vMF & 32 & .2124 & .2199 & .2199 & .2218 & .2221 & .2228 \\
 & Pure-vMF & 64 & .2218 & .2228 & .2233 & .2232 & .2229 & .2228 \\
\midrule
\multirow{17}{*}{NQ} & No DP & $\infty$ & .5360 & .5391 & .5392 & .5401 & .5402 & .5414 \\
 & Randomized response & 8 & .0000 & .0003 & .0010 & .0035 & .0039 & .5414 \\
 & Randomized response & 16 & .0005 & .0021 & .0032 & .0073 & .0086 & .5414 \\
 & Randomized response & 32 & .0038 & .0071 & .0118 & .0230 & .0304 & .5414 \\
 & Randomized response & 64 & .0258 & .0417 & .0574 & .0789 & .0932 & .5414 \\
 & Gaussian & 8 & .4765 & .4988 & .5107 & .5217 & .5247 & .5414 \\
 & Gaussian & 16 & .5253 & .5322 & .5364 & .5376 & .5385 & .5414 \\
 & Gaussian & 32 & .5341 & .5369 & .5379 & .5396 & .5400 & .5414 \\
 & Gaussian & 64 & .5348 & .5375 & .5387 & .5392 & .5396 & .5414 \\
 & RDP-vMF & 8 & .4911 & .5093 & .5174 & .5251 & .5278 & .5414 \\
 & RDP-vMF & 16 & .5305 & .5366 & .5385 & .5393 & .5390 & .5414 \\
 & RDP-vMF & 32 & .5342 & .5380 & .5387 & .5398 & .5400 & .5414 \\
 & RDP-vMF & 64 & .5338 & .5371 & .5387 & .5395 & .5401 & .5414 \\
 & Pure-vMF & 8 & .0238 & .0345 & .0503 & .0690 & .0837 & .5414 \\
 & Pure-vMF & 16 & .2158 & .2679 & .3127 & .3496 & .3758 & .5414 \\
 & Pure-vMF & 32 & .4746 & .4969 & .5088 & .5206 & .5252 & .5414 \\
 & Pure-vMF & 64 & .5231 & .5317 & .5343 & .5372 & .5389 & .5414 \\
\midrule
\multirow{17}{*}{FEVER} & No DP & $\infty$ & .8448 & .8470 & .8480 & .8483 & .8488 & .8495 \\
 & Randomized response & 8 & .0001 & .0006 & .0011 & .0017 & .0021 & .8495 \\
 & Randomized response & 16 & .0003 & .0010 & .0017 & .0042 & .0055 & .8495 \\
 & Randomized response & 32 & .0025 & .0055 & .0088 & .0160 & .0198 & .8495 \\
 & Randomized response & 64 & .0258 & .0430 & .0619 & .0888 & .1035 & .8495 \\
 & Gaussian & 8 & .7609 & .7884 & .8046 & .8188 & .8236 & .8495 \\
 & Gaussian & 16 & .8398 & .8442 & .8458 & .8472 & .8479 & .8495 \\
 & Gaussian & 32 & .8450 & .8463 & .8476 & .8483 & .8488 & .8495 \\
 & Gaussian & 64 & .8463 & .8474 & .8478 & .8482 & .8487 & .8495 \\
 & RDP-vMF & 8 & .7942 & .8129 & .8237 & .8336 & .8370 & .8495 \\
 & RDP-vMF & 16 & .8422 & .8451 & .8469 & .8483 & .8486 & .8495 \\
 & RDP-vMF & 32 & .8453 & .8473 & .8479 & .8483 & .8485 & .8495 \\
 & RDP-vMF & 64 & .8459 & .8477 & .8479 & .8482 & .8485 & .8495 \\
 & Pure-vMF & 8 & .0184 & .0307 & .0437 & .0647 & .0806 & .8495 \\
 & Pure-vMF & 16 & .2781 & .3540 & .4133 & .4766 & .5186 & .8495 \\
 & Pure-vMF & 32 & .7579 & .7875 & .8027 & .8177 & .8248 & .8495 \\
 & Pure-vMF & 64 & .8359 & .8433 & .8450 & .8466 & .8471 & .8495 \\
\bottomrule
\end{tabular}
\end{table*}%
}

\title{Pointing the Way, Hiding the Destination: Practical Private Dense Retrieval at Scale}

\author{
\normalfont
Peichun Hua\textsuperscript{1,2},
Danyang Chen\textsuperscript{1},
Junan Zhang\textsuperscript{1},
Haifeng Sun\textsuperscript{3},
Jingyu Wang\textsuperscript{3},\\ \normalfont
Diwen Xue\textsuperscript{4},
Mingyu Li\textsuperscript{5},
Yunming Xiao\textsuperscript{1,2}\thanks{Corresponding author.}\\[0.8em]
\textsuperscript{1}The Chinese University of Hong Kong, Shenzhen
\\
\textsuperscript{2}State Key Laboratory of Internet Architecture, Tsinghua University \\
\textsuperscript{3}Beijing University of Posts and Telecommunications \quad
\textsuperscript{4}The Chinese University of Hong Kong\\
\textsuperscript{5}Institute of Software, Chinese Academy of Sciences
}

\begin{document}

\thispagestyle{plain}
\pagestyle{plain}

\maketitle

\begin{abstract}
Hosted retrieval-augmented generation (RAG) and semantic search allow users to query valuable provider-held corpora, raising two competing demands: to hide each query and chosen result, yet reveal only the documents that the user is authorized to receive.
Existing cryptographic approaches either make this costly by processing the entire corpus for every query, or sacrifice quality for efficiency by scanning a few clusters.
We repurpose learned deep hashing as a private filter: a randomized binary code points the provider to a short candidate list, while encrypted reranking and oblivious key transfer protect the precise query and final selection. This shortlist short-circuits full-corpus cryptographic search without sacrificing retrieval quality: with 200–500 candidates, it closely matches full-corpus retrieval across five zero-shot corpora spanning 25K to 5.4M documents.
On the full 2.68M-passage NQ corpus over a 10-Gbps link, our protocol only adds 0.38 seconds, or 6.5\%, to a 128-token Qwen3-32B RAG pipeline. The released code satisfies directional metric differential privacy (DP) and substantially reduces embedding-inversion and property-inference leakage, demonstrating that a carefully learned shortlist can make private dense retrieval both accurate and practical.
\end{abstract}

\section{Introduction}
\label{sec:intro}

The proliferation of Retrieval-Augmented Generation (RAG) has transformed large language models (LLMs) from static artifacts into dynamic, knowledge-grounded reasoning engines \cite{fan2024survey,gao2023retrieval,ram2023context,lewis2020rag,hua2026semantics}. At the core of RAG lies dense retrieval: mapping documents and queries into a high-dimensional vector space and searching via nearest neighbors. As organizations deploy RAG over sensitive corpora such as legal archives, medical records, and proprietary knowledge bases, both the document collection and user queries become sensitive assets~\cite{tang2026dprag,chang2025remoterag,yao2026connect,nguyen2026five}. Recent embedding inversion attacks \cite{song2020information,morris2023text,li2023geia,huang2024transferable,chen2024text,zhang2025universal} have demonstrated that continuous embeddings are \textit{not} opaque fingerprints, but rich, invertible semantic representations from which an adversary can reconstruct sensitive source text and infer private attributes.

Existing cryptographic approaches to private retrieval struggle to balance security with the scale and dimensionality of modern dense embeddings. Homomorphic encryption (HE) and multi-party computation (MPC) protocols \cite{chang2025remoterag,li2024avpmir,ming2026p2rag} incur orders-of-magnitude of computational overhead, while ORAM-based methods \cite{zhu2025compass} require multiple interactions and high user-side computation and storage~\cite{cui2026mess}.
Trusted execution environments (TEEs) avoid homomorphic scoring,
while access-pattern~\cite{xu2015controlled} and microarchitectural leakage~\cite{vanbulck2018foreshadow} remain.

\textbf{This work.}
Our central design choice is to expose a coarse candidate pattern under directional metric differential privacy (mDP)~\cite{imola2024metric} and reserve cryptographic computation for the resulting shortlist. An honest-but-curious \emph{Corpus Owner} keeps its proprietary corpus with full-precision embeddings while additionally maintaining a lightweight binary index; an authorized \emph{User} randomizes released hash codes that satisfy stringent mDP guarantees and encrypts the clean query representation. The resulting code makes the nearby query directions induce similar candidate-pattern distributions, making them hard to distinguish from one another and thereby protecting critical detailed privacy information (\S\ref{sec:attacks}).
The \textit{Owner} performs a Hamming search on the corpus ($N$) and packed Brakerski–Fan–Vercauteren (BFV) homomorphic scoring over the resulting $K$ candidates. The \textit{User} decrypts the scores, selects locally, and uses active-secure $k$-out-of-$K$ oblivious transfer (OT) to open at most $k$ payloads without revealing its choices or accessing the other $K-k$ payloads.

Recovering retrieval quality under DP requires scoring $K=500$--$3{,}000$ candidates, while a RAG request typically releases and bills only $k=3$--$10$ documents.
The gap between $K$ and $k$ makes OT central to this deployment model. Returning the shortlist would disclose hundreds of times more content; requesting the final documents directly would reveal the User's selection.
Instead, our $k$-out-of-$K$ transfer binds each round to the billable result count, while authenticated accounts and cumulative quotas govern repeated extraction.
Our model also covers a \textit{User} that actively tries to learn payloads beyond its billed $k$ per round, and the protocol bounds this extraction (\S\ref{sec:user-access}; Theorem~\ref{thm:payload-access}).

Three empirical insights make our design highly practical.

\noindent\ding{182}\enspace\textbf{Candidate-set redundancy persists across scales and domains.} A learned binary filter (Stage~1) needs only to preserve the documents that matter for the final ranking (Stage~2), while the ranking among the top documents is not critical. At $K=500$, the two-stage pipeline (\S\ref{sec:architecture}) retains 98.84\%--100\% of full-corpus NDCG@10 (ranking quality) across five zero-shot BEIR corpora~\cite{thakur2021beir} on both evaluated encoders. In particular, on Climate-FEVER with 5.4M documents, the shortlist excludes noisy high-scoring documents and raises NDCG@10 by 0.0012–0.0158.
Even under DP protection, expanding the candidate pool up to 3,000 achieves a good balance between efficiency and search quality (\S\ref{sec:search_quality}).
These corpora span scientific search, question answering, entity retrieval, and fact verification at scales from 25K to 5.4M documents (Figure~\ref{fig:k-scaling}), demonstrating strong generalizability of our approach.

\noindent\ding{183}\enspace\textbf{Dense embeddings contain substantial precision redundancy.} We find that scalar int8 quantization preserves the pretrained ranking closely; symmetric quantization converts floating-point similarity into exact, low-depth integer dot products. Without this precision reduction, private real-valued scoring relies on approximate-number HE such as CKKS~\cite{cheon2017homomorphic}, higher-precision fixed-point BFV~\cite{juvekar2018gazelle}, or interactive fixed-point MPC~\cite{mohassel2017secureml}; int8 enables exact packed BFV scoring with a small plaintext modulus while preserving high ranking quality.

\noindent\ding{184}\enspace\textbf{Model separation creates pipeline parallelism.} We adopt low-rank adaptation (LoRA) fine-tuning~\cite{hu2022lora} of pretrained encoders to create the learned hash models efficiently. This ensures that the hash codes can select a high-recall subset without incurring the storage and memory burden of an entirely new model. Our protocol implementation carefully overlaps the pretrained encoder pass and BFV query encryption with Owner-side Hamming shortlist construction, then streams the candidate ciphertexts while encrypted scoring and key transfer advance. This successfully hides the latency and makes our protocol efficient under different network bandwidths.

\begin{figure}[!b]
\centering
\includegraphics[width=0.7\columnwidth]{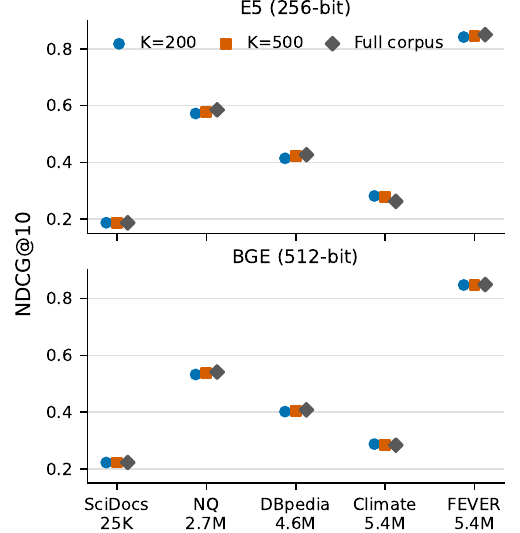}
\caption{Stage~2 NDCG@10 at $K\in\{200,500\}$ versus full-corpus retrieval on five BEIR corpora (25K--5.4M docs).}
\label{fig:k-scaling}
\end{figure}

We investigate randomized response, Gaussian, R\'enyi-DP-calibrated von Mises–Fisher (RDP-vMF), and exactly calibrated pure-vMF mechanisms that achieve mDP, extensively evaluating their privacy–utility trade-offs through retrieval quality, representation attacks, protocol latency, and end-to-end RAG.
In summary, we make the following contributions:
\begin{enumerate}
\item \textbf{A leakage-aware retrieval protocol and security contract} that reveals a metric-DP candidate pattern, confines encrypted exact scoring to $K$ candidates, protects the clean query against the Owner, and uses active-secure $k$-out-of-$K$ OT to hide the final selection and bound payload recovery.
\item \textbf{A learned candidate-filter design and training recipe} that retrofits pretrained encoders via LoRA while preserving the original model as the Stage~2 scorer. It sustains near-lossless zero-shot retrieval across domains and corpus scales and improves over direct quantization, classical hashing, and supervised binary-retrieval baselines.
\item \textbf{A directional randomization mechanism and analysis} that gives the released shortlist code pure metric-DP protection, carries the guarantee through candidate generation by post-processing, and establishes its retrieval--privacy frontier against randomized response, Gaussian, and RDP-vMF mechanisms and representation attacks.
\item \textbf{A practical two-stage system realization} that combines int8 quantization, low-depth packed BFV, OT over content keys only, and pipeline overlap in a networked prototype, reducing cryptographic work from $N$ documents to a high-recall shortlist while retaining efficient million-scale retrieval and end-to-end RAG.
\end{enumerate}
 
\section{Background and Related Work}
\label{sec:background}

\subsection{Dense Retrieval and Learned Hashing}

RAG grounds language-model generation in retrieved evidence~\cite{fan2024survey,gao2023retrieval,ram2023context}. Its scalable first stage is usually a bi-encoder such as DPR, SBERT, E5, or BGE~\cite{karpukhin2020dpr,reimers2019sentence,wang2022e5,chen2024bge}: queries and documents are encoded independently, then compared by inner product or cosine similarity~\cite{reimers2019sentence}. Cross-encoders, instead, score query–document pairs jointly and are commonly reserved for reranking because their cost grows with the number of evaluated pairs~\cite{chiu2022cross,chen2024how,rauf2024bce}.

Binary hashing compresses continuous representations into $L$-bit codes and replaces floating-point distance with XOR and popcount~\cite{luo2023survey}. Classical methods include data-independent random-hyperplane, multi-probe, and cross-polytope LSH variants~\cite{charikar2002similarity,lv2007multiprobe,andoni2015practical,pham2022falconn}, and data-dependent rotations such as ITQ and IsoHash~\cite{gong2011itq,kong2012isotropic}. Deep models jointly learn the representation and code: earlier work emphasized image retrieval with pairwise or center-based objectives~\cite{li2015feature,zhu2016deep,liu2019deep,yuan2020central,wang2023deep}, ranking-aware hashing directly optimized tie-aware AP and NDCG in Hamming space~\cite{he2018hashing}, and recent objectives couple discrimination with quantization or adapt transformer encoders to text~\cite{hoe2021one,ou2021dhim,he2025survey}.
Instead of preserving embedding geometry and search quality, we use a learned code only as a high-recall first-stage filter and retain the continuous representation for exact candidate reranking. This division gives the Owner a lightweight coarse index while reducing private similarity evaluation from all $N$ documents to $K$ candidates.

\subsection{Representation Privacy}

Dense embeddings preserve enough lexical and semantic information to support generation-based, search-based, and property-inference attacks~\cite{song2020information,morris2023text,li2023geia,huang2024transferable,chen2024text,zhang2025universal}. Our design randomizes the normalized pre-binarization representation under metric DP~\cite{imola2024metric,xie2025decade} and then deterministically binarizes it; the released code inherits the same privacy bound by post-processing~\cite{dwork2014foundations}.

Dense retrievers normalize embeddings and rank them by cosine similarity, so their semantic neighborhoods lie naturally on the unit sphere and are measured by angle. We therefore instantiate metric DP over nearby directions of the learned coarse representation~\cite{dwork2006calibrating,dwork2014foundations,chatzikokolakis2013broadening}.

\begin{definition}[$(\varepsilon,\delta)$-Directional Privacy]
\label{def:directional-privacy}
Let $d$ be a metric on the unit sphere and let $\rho>0$. A randomized mechanism $\mathcal{M}:\mathbb{S}^{L-1}\rightarrow\mathcal{Y}$ satisfies $(\varepsilon,\delta)$-directional privacy at radius $\rho$ if, for every $u,u'\in\mathbb{S}^{L-1}$ \underline{with $d(u,u')\le\rho$} and every measurable $S\subseteq\mathcal{Y}$,
\begin{equation}
\Pr[\mathcal{M}(u)\in S]\le e^{\varepsilon}\Pr[\mathcal{M}(u')\in S]+\delta.
\end{equation}
\end{definition}

This \underline{radius-bounded} definition instantiates metric privacy~\cite{chatzikokolakis2013broadening,xie2025decade} on normalized directions. We use angular distance $d_\theta(u,u')=\arccos(u^\top u')$ as the primary metric and chord distance $d_c(u,u')=\|u-u'\|_2=2\sin(d_\theta(u,u')/2)$ when deriving the vMF density-ratio bound. The radius defines a neighborhood in representation space. For $L=256$, the angular radius $\rho_\theta=0.125$ used in our evaluation is the angle between two saturated pre-sign directions that differ in one bit, and metric DP scales protection gracefully with distance: under the pure-vMF release (\S\ref{sec:directional-release}), saturated directions that differ in $m$ signs are protected at $\varepsilon\sqrt{m}$. Queries that differ only in a detail a few bits wide therefore receive close to the per-bit budget ($m\le4$ costs at most $2\varepsilon$), while queries on different topics, which must retrieve different documents, remain distinguishable. Hiding the topic as well requires cryptography over the whole corpus, whose cost grows with the corpus size and is significantly higher (\S\ref{sec:efficiency}).

Prior private-hashing methods randomize discrete codes with randomized response or randomize data-independent LSH functions~\cite{wang2020searching,fernandes2021dplsh,cui2026mess}. Deep hashing instead exposes the continuous pre-sign representation as a natural randomization point, preserving coordinate margins and directional geometry. Gaussian perturbation operates on bounded $h\in[-1,1]^L$ under Euclidean adjacency~\cite{balle2018gaussian}, while von Mises--Fisher (vMF) perturbation operates on normalized $u=h/\|h\|_2$ under angular or chordal adjacency~\cite{weggenmann2021directional}; binarization then preserves their guarantees by post-processing~\cite{dwork2014foundations}.

Randomized response treats every bit alike and must compose privacy across the complete code because nearby continuous vectors can cross many low-margin sign boundaries. At $L=256$ and $(\varepsilon=16,\delta=10^{-6})$, this calibration flips 46.1\% of the bits and yields 0.0306 mean NDCG@10 at $K=3000$, versus 0.536 for Gaussian and 0.537 for RDP-vMF (Table~\ref{tab:dp-retrieval}).

\subsection{Private Retrieval across Deployments}

The ownership and trust boundary define the private-retrieval problem in different scenarios~\cite{yaluhin2026sok}. Table~\ref{tab:private-retrieval-systems} organizes prior systems into four deployment models. 

In \emph{client-owned outsourcing}, classical searchable encryption~\cite{curtmola2006sse,chen2018dpaccess,shang2021osse} and vector systems~\cite{liu2025ppann,zhu2025compass,cui2026mess,hua2026spruce} protect a corpus owner who queries an untrusted cloud. For example, MESS~\cite{cui2026mess} randomizes LSH codes, maintains 64 HNSW shards with each item routed to 16, and reranks the aggregated candidates at the client; this yields 16$\times$ indexed-record replication. Retrieving the original documents from the cloud additionally requires private information retrieval (PIR)~\cite{chor1998private,ulukus2022private} to hide the access pattern.

In the \emph{public/shared-corpus} model, Tiptoe, Wally, PACMANN, and Speakeasy protect the query without a corpus-confidentiality goal~\cite{henzinger2023tiptoe,asi2024wally,zhou2025pacmann,lakshman2026speakeasy}. For example, Tiptoe~\cite{henzinger2023tiptoe} combines clustering with linearly homomorphic search across 45 servers to operate vector search on web-scale corpus, but its single-cluster pruning for efficiency significantly impacts the search quality compared to normal embedding search. PACMANN~\cite{zhou2025pacmann} improves this quality--latency tradeoff through graph search and client-preprocessing PIR, but with 100 million vectors, each client downloads 59.6~GB during setup, stores 2.9~GB of state, and exchanges another 399.4~MB to maintain that state after every query.

In \emph{secret-shared outsourcing}, PRAG and P$^2$RAG protect distinct data owners and queriers by placing the database across MPC servers~\cite{zyskind2024prag,ming2026p2rag}; PRAG assumes an honest majority, whereas P$^2$RAG uses two semi-honest non-colluding servers and avoids secure sorting through interactive bisection, but requires full-corpus work and trusted-dealer preprocessing, the cost of which is excluded from its benchmark.

Our target is the fourth model, a \emph{provider-held proprietary corpus} serving an external querier~\cite{servanschreiber2022private,chen2020sanns,li2025panther,chang2025remoterag,liang2026pisces}. Pisces~\cite{liang2026pisces} combines an oblivious SimHash filter with MPC scoring and PIR, and adds a cryptographic BM25 path; our learned filter releases a directionally private shortlist and concentrates cryptography on exact candidate scoring and selected payloads. PANTHER~\cite{li2025panther} co-designs PIR, secret sharing, garbled circuits, and homomorphic encryption for strong single-provider protection, but its cluster-wise PIR representation scales with both vector dimensionality and the number of probes. Its evaluation targets 96- and 128-dimensional vision embeddings, whereas modern RAG commonly uses substantially wider text embeddings and demands enough probes to preserve near-lossless retrieval quality. In our evaluation (\S\ref{sec:efficiency}), PANTHER exceeds the memory limit of our server in a million-scale corpus. Concurrent work by Tremblay Thibault et al.~\cite{thibault2026quantization} answers each query with verifiable homomorphic matrix--vector products over committed low-bit-quantized matrices, adding result integrity against a malicious provider. Its cryptographic work, however, stays linear in $N$ and takes minutes per query at $10^6$ chunks without protecting the corpus from the client, whereas our shortlist confines scoring to $K$ candidates and limits a malicious User to $k$ payloads. Such verifiable retrieval, including Atlas's zero-knowledge proofs of HNSW search over cleartext queries~\cite{avramov2026atlas}, is complementary and could harden our candidate scoring against a deviating Owner. Within this model, our protocol keeps the corpus at a single provider for the owner, provides formal DP guarantees for the search pattern, cryptographically protects unselected corpus content from the user, and confines homomorphic scoring to the filtered candidates for efficiency and near-lossless quality.

\begin{table*}[t]
\centering
\caption{Private retrieval systems grouped by corpus ownership, querier role, and trusted infrastructure.}
\label{tab:private-retrieval-systems}
\footnotesize
\setlength{\tabcolsep}{2.5pt}
\renewcommand{\arraystretch}{1.00}
\begin{tabular}{@{}p{3.30cm}p{2.05cm}p{0.8cm}*{5}{>{\centering\arraybackslash}p{1.10cm}}p{2.90cm}@{}}
\toprule
Scheme & Search & Infra. & \shortstack{Query\\privacy} & \shortstack{Content\\privacy} & \shortstack{Pattern\\privacy} & \shortstack{Near-\\lossless} & \shortstack{Fast\\online} & Index / auxiliary state \\
\midrule
\multicolumn{9}{@{}l}{\textit{Client-owned outsourcing---corpus owner is querier; cloud is adversary}} \\
CGKO06~\cite{curtmola2006sse} & Keyword & 1S & \yesmark & -- & \nomark & -- & -- & Inverted index \\
CLRZ18~\cite{chen2018dpaccess}, SOPK21~\cite{shang2021osse} & Keyword & 1S & \yesmark & -- & \dpmark & -- & -- & DP index \\
LZXL25~\cite{liu2025ppann} & Graph & 1S & \partialmark & -- & \nomark & \partialmark & \yesmark & Graph + 2 CT \\
Compass~\cite{zhu2025compass} & Graph & 1S & \yesmark & -- & \yesmark & \yesmark & \partialmark & 3.2--6.8$\times$ server \\
MESS~\cite{cui2026mess} & Multi-graph & 1S & \dpmark & -- & \dpmark & \partialmark & \partialmark & 16$\times$ HNSW index\\
\addlinespace[1pt]
\multicolumn{9}{@{}l}{\textit{Public/shared corpus---no corpus-confidentiality goal}} \\
Tiptoe~\cite{henzinger2023tiptoe} & $k$-means & 45S & \yesmark & -- & \partialmark & \nomark & \nomark & Cluster index \\
Wally~\cite{asi2024wally} & $k$-means & Crowd & \dpmark & -- & \dpmark & \nomark & \nomark & Cluster index \\
PACMANN~\cite{zhou2025pacmann} & Graph & 1S+P & \yesmark & -- & \partialmark & \partialmark & \nomark & Graph + client hints \\
\addlinespace[1pt]
\multicolumn{9}{@{}l}{\textit{Secret-shared outsourcing---distinct owner and querier}} \\
PRAG~\cite{zyskind2024prag} & IVF & MPC & \yesmark & \yesmark & \yesmark & \partialmark & \nomark & Database shares \\
P$^2$RAG~\cite{ming2026p2rag} & Full scan & 2NC & \yesmark & \partialmark & \yesmark & \yesmark & \nomark & Two DB shares \\
\addlinespace[1pt]
\multicolumn{9}{@{}l}{\textit{Provider-held proprietary corpus---external querier}} \\
Pisces~\cite{liang2026pisces} & \shortstack{SimHash + BM25} & 1S & \yesmark & \yesmark & \partialmark & \partialmark & \nomark & $160N$-CT + token OKVS \\
SANNS~\cite{chen2020sanns} & $k$-means & 1S & \yesmark & \yesmark & \yesmark & \partialmark & \nomark & DORAM \\
PANTHER~\cite{li2025panther} & $k$-means & 1S & \yesmark & \yesmark & \yesmark & \partialmark & \nomark & PIR + MPC \\
RemoteRAG~\cite{chang2025remoterag} & ANN & 1S & \dpmark & \partialmark & \dpmark & \yesmark & \partialmark & ANN index \\
TAKA26~\cite{thibault2026quantization} & Full scan & 1S & \yesmark & \nomark & \yesmark & \partialmark & \nomark & Matrix commitments \\
\textbf{Ours} & Hash scan & 1S & \dpmark & \textbf{\yesmark} & \dpmark & \textbf{\yesmark} & \textbf{\yesmark} & \textbf{32 B/doc filter} \\
\bottomrule
\end{tabular}
\vspace{1pt}
\parbox{0.985\textwidth}{\scriptsize \textbf{Legend.} \yesmark: cryptographic protection or full support; \dpmark: formal DP guarantee; \partialmark: empirical or partial; \nomark: unsupported; --: inapplicable. Content privacy limits payload recovery to authorized results; pattern privacy hides whether queries repeat (search) and which items are touched (access). Near-lossless: exact or $\ge$99\% of plaintext quality. Last column: native search structure beyond the embeddings (ours: the 256-bit index). 1S: one server; 2NC: two non-colluding servers; CT: ciphertext; OKVS: oblivious KV store; P: PIR preprocessing.}
\end{table*}

\section{Deployment and Threat Model}
\label{sec:protocol}

We target a proprietary corpus served directly by its Owner to an external authorized User.

\subsection{Parties and Trust Relations}
\label{sec:parties}

\textbf{Corpus Owner (Server).} Holds the document collection, the binary hash index, the normalized embeddings used as plaintext HE operands, per-document content keys, and document payloads protected by authenticated encryption with associated data (AEAD). The Owner is \emph{honest-but-curious}: it follows the prescribed computation and message schedule while attempting to infer query content, link queries, or profile Users from its protocol view. All server-side retrieval runs on Owner-controlled infrastructure.

\textbf{User (Client).} Holds a query, the agreed encoder and the hash model, the HE secret/public keys, and the DP parameters. The HE scoring guarantee applies to a conforming User that encrypts the prescribed bounded, canonically packed query; the active-secure OT guarantee additionally covers a malicious receiver attempting to recover more than $k$ content keys.

\subsection{Views, Leakage, and Assumptions}
\label{sec:threat}

The protocol has four security goals. Metric DP protects the coarse query code released to the Owner, BFV semantic security protects the clean query from the Owner, ciphertext-simulatable BFV restricts a conforming User's scoring view to the prescribed $K$ exact scores, and active-secure OT hides the User's selected positions while limiting payload-key recovery to $k$ OT choices.

The Owner observes:
\begin{itemize}
\item Binary index $\{\bcode_d\}_{d \in [N]}$ and the resulting candidate set $C_K$;
\item Metric-DP coarse query code $\tilde{\bcode}_q$;
\item The HE ciphertext $\mathsf{Enc}(\bar{\mathbf{q}})$;
\item The HE and sender-side OT transcripts;
\item Authenticated identity, session and round identifiers, message lengths and timing, $K$, and the public result count $k$.
\end{itemize}
The Owner does \emph{not} observe the User's plaintext query, the decrypted similarity scores, or which specific $k$ indices the User selected via OT.

The User observes:
\begin{itemize}
\item The $K$ scalar similarity scores;
\item The $K$ AEAD payload ciphertexts and the $k\times K$ masked content-key table for the candidate set;
\item At most $k$ content keys and the corresponding plaintext document payloads.
\end{itemize}
Under conforming execution, the application gives the User candidate-local scores and selected payloads while keeping the binary index, corpus embeddings, document identifiers, and unselected content keys at the Owner. Randomized evaluation of the compact layout makes the evaluated ciphertext simulatable from the query ciphertext, the $K$ scores, and public metadata, so its other slots and coefficients add no corpus-embedding information beyond this explicit score oracle (\S\ref{sec:bfv-correctness}). We assume an authenticated encrypted channel (e.g., TLS), which reveals only message lengths and timing.

\emph{Assumptions and scope.} A trusted model-distribution step fixes the encoder, hash head, quantization parameters, and HE parameters shared by both parties. The HE scoring theorems apply to fresh symmetric BFV encryptions of bounded, canonically packed queries; the score quota governs cumulative exposure but does not establish consistency between the encrypted query and the coarse code. The guarantees assume uncompromised endpoints, authenticated identities, and authenticated confidential transport, with message lengths and timing treated as explicit leakage. They cover query privacy, exact and ciphertext-simulatable scoring for conforming inputs, and per-round payload access; Sybil resistance, availability, inference from the released exact scores, and malformed-ciphertext server privacy remain outside the model.

\subsection{Empirical Privacy Attacks}
\label{sec:attacks}

We evaluate three attacks that recover text or sensitive attributes from an exposed representation. DP mechanism experiments target the released Stage~1 code, while the no-DP learned code and float embedding provide reference points.

\textbf{Search-based inversion.} ZSInvert~\cite{zhang2025universal} treats reconstruction as black-box optimization. Given a target representation $r$, an LLM proposes a beam of candidate texts, the target encoder maps each candidate to the same representation space, and similarity to $r$ selects the next beam. The best candidate then seeds another refinement round. Float targets use cosine similarity, whereas binary targets use normalized hash similarity; the latter supplies only $L+1$ distinct Hamming-similarity values. The output is the highest-scoring reconstructed text.

\textbf{Generation-based inversion.} GEIA~\cite{li2023geia} learns an embedding-conditioned autoregressive decoder from auxiliary text--representation pairs. A learned projection maps the target representation into the decoder input space, and teacher-forced language-model training maximizes $\prod_t p_\phi(x_t\mid x_{<t},r)$. At inference time, the trained decoder reconstructs each held-out target in one generation pass.

\textbf{Property inference.} The attacker obtains an auxiliary set of representation--attribute pairs~\cite{song2020information}, trains a classifier to predict the attribute from the exposed representation, and applies the selected classifier to held-out victim representations. We evaluate topic, sentiment, and authorship because they span coarse semantic content, affect, and fine-grained source identity. It needs no reconstruction of the text.

These attacks measure representation leakage (\S\ref{sec:security_eval}); \S\ref{sec:threat} accounts for candidate identities and cross-round linkage.

\section{Deep Hash Learning}
\label{sec:hashing}

The learned hash encoder turns a pretrained dense retriever into a high-recall candidate filter, while the original pretrained encoder remains the Stage~2 scorer. We focus here on the model architecture and the training signals that make this separation effective; Appendix~\ref{app:setup} gives the exact losses, discretization schedule, and optimization parameters.

\subsection{Motivation and Architecture}
\label{sec:hashing_motivation}
A key challenge in deep hashing is preserving zero-shot candidate recall after adapting a pretrained encoder to a discrete space~\cite{he2025survey}. For text $x$, our hash model applies a linear head to a LoRA-adapted encoder~\cite{hu2022lora} and emits
\begin{equation}
\bcode(x)=\sign\!\left(W\,\mathrm{pool}(E_{\mathrm{LoRA}}(x))\right)\in\{-1,1\}^{L}.
\label{eq:hash-model}
\end{equation}
Here $W\in\mathbb{R}^{L\times d}$, where $d$ is the encoder hidden dimension and $L$ is the code length. We use mean pooling and $L=256$ for E5-base-v2~\cite{wang2022e5}, and [CLS]-token pooling and $L=512$ for BGE-base-en-v1.5~\cite{chen2024bge}. The linear head and compact LoRA update specialize the pretrained representation for Hamming candidate recall without training another backbone from scratch.

Candidate generation and scoring use separate model states. The adapted encoder and hash head produce only the coarse code, while the unchanged pretrained encoder supplies the continuous query and document embeddings used by Stage~2. This separation lets Stage~1 reshape its geometry for Hamming search without shifting the final dense ranking. Online, both encoders share one tokenization, and the pretrained encoder pass overlaps Owner-side Hamming search (\S\ref{sec:online}).

\subsection{Encoder Tuning}
\label{sec:tuning}

We train the hash model on MS MARCO query--passage supervision using three functional groups:
\begin{equation}
\mathcal{L}=\mathcal{L}_{\mathrm{retrieval}}+\mathcal{L}_{\mathrm{transfer}}+\mathcal{L}_{\mathrm{regularization}}.
\label{eq:hash-objective-groups}
\end{equation}
Direct retrieval supervision brings relevant query--passage pairs together and pushes mined negatives away in the adapted continuous space; E5 additionally applies this supervision directly in Hamming space. Ranking transfer carries the adapted encoder's ordering into the deployed Hamming space. The remaining regularizers anchor the adapted representation to the frozen pretrained geometry and keep examples well spread before binarization. Appendix~\ref{app:setup} defines every component of Equation~\ref{eq:hash-objective-groups} and reports its model-specific weight.

\subsection{Hard-Negative Training}
\label{sec:hashing_optimization}
A hard negative is a non-relevant passage that remains deceptively close to the query, making it more informative than a random passage for learning the candidate boundary. We mine these examples only from the MS MARCO training corpus: a broad lexical retrieval stage finds plausible candidates, a stronger reranker orders them, and training samples from the highest-ranked non-relevant passages. The reranker also suppresses likely unlabeled positives so that ambiguous passages do not become contradictory supervision. This source-only procedure teaches the hash model to preserve fine distinctions without adapting to any evaluation corpus.

E5 and BGE each train for 16 epochs. Training moves progressively from a smooth representation to the binary codes used at deployment; Appendix~\ref{app:setup} specifies this schedule together with the mining models, sample counts, learning rates, and remaining hyperparameters.



\section{DP-Filtered Private Dense Retrieval}
\label{sec:architecture}

We now describe the complete two-party protocol. Throughout, $K$ denotes the candidate budget (the number of documents that receive HE scoring) and $k$ denotes the final result count returned to the User.

\subsection{Offline Setup}
\label{sec:setup}

\textbf{Owner setup.} The Owner runs two document encoders offline. The LoRA-adapted hash encoder and linear hash head produce $\bcode_i\in\{-1,1\}^L$ for candidate generation, while the unchanged pretrained encoder produces the normalized scoring embedding $\mathbf{z}_i\in\mathbb{R}^d$. A shared symmetric quantizer with zero point 0 and scale $a=\max_{i,j}|z_{i,j}|/127$ maps the pretrained embeddings to $\{-127,\ldots,127\}^d$; the symmetric range excludes $-128$ and bounds every integer dot product by $127^2d$. The Owner stores:
\begin{itemize}
\item A flat binary index $\{\bcode_i\}_{i=1}^N$ for Hamming-distance search;
\item The quantized embeddings $\{\bar{\mathbf{z}}_i\}_{i=1}^N$ as Owner-local plaintext HE operands;
\item A random 128-bit content key $\kappa_i$ and an AEAD ciphertext $P_i$ for each document. The plaintext contains its true-length field and is zero-padded to a 4096-byte boundary before one-shot ChaCha20--Poly1305 encryption under a 256-bit key derived from $\kappa_i$ with HKDF--SHA-256, so ciphertext length reveals only the padded block count.
\end{itemize}

\textbf{User and session setup.} For encrypted candidate scoring, we 
instantiate single-instruction multiple-data (SIMD) batched BFV for the quantized integer dot products~\cite{fan2012somewhat,homomorphicEncryptionStandard2018}, packing multiple score lanes into each ciphertext. The User obtains the LoRA-merged hash encoder, hash head, original pretrained encoder, tokenizer, and quantizer. It generates the BFV secret/public keys and required Galois keys, retains the secret key, and sends only public and evaluation material to the Owner. An authenticated session binds an identity, layout, $K$, $k$, and monotone round counter. The two parties establish the base-OT correlation state once per long-lived OT connection; each query advances the extension state to derive fresh rows rather than repeating base OT.

\subsection{Online Protocol}
\label{sec:online}

Figure~\ref{fig:protocol-overview} summarizes the eight online stages described below; Appendix~\ref{app:protocol-sequence} provides the complete message sequence.

\begin{figure}[t]
\centering
\includegraphics[width=\linewidth]{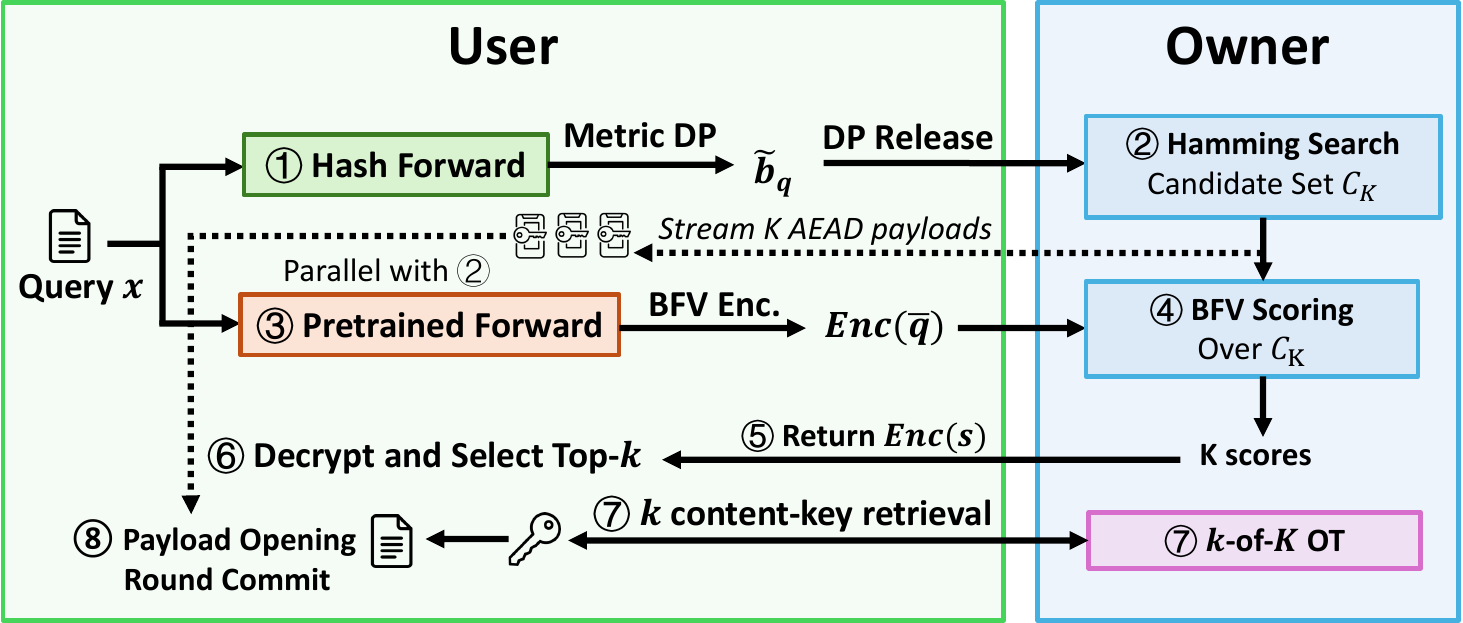}
\caption{\textbf{Online retrieval pipeline.} Metric-DP Hamming shortlisting overlaps the pretrained encoder pass and query encryption. Once $C_K$ is fixed, its AEAD ciphertexts stream during exact BFV scoring.}
\label{fig:protocol-overview}
\end{figure}

For each query $x$, the following steps are executed:

\begin{enumerate}
\item \textbf{Hash encoding and DP release.} The User tokenizes $x$ and runs the LoRA-merged hash encoder, whose head outputs logits $\mathbf{g}_q\in\mathbb{R}^L$.
Then, using the final training scale $\beta$, the User computes the bounded pre-binarization vector $\mathbf{h}_q^{\text{soft}}=\tanh(\beta\mathbf{g}_q)$, normalizes it to $\mathbf{u}_q=\mathbf{h}_q^{\text{soft}}/\|\mathbf{h}_q^{\text{soft}}\|_2$, samples $\mathbf{y}_q\sim\mathrm{vMF}(\mathbf{u}_q,\kappa)$ with the pure calibration in Equation~\ref{eq:vmf-calibration}, and sends $\tilde{\bcode}_q=\sign(\mathbf{y}_q)$.

\item \textbf{Hamming search and payload streaming.}  The Owner starts Hamming search over $\{\bcode_i\}_{i=1}^N$. Once $C_K=(i_1,\ldots,i_K)$ is fixed, it reserves the User's query and score quota and immediately streams the AEAD ciphertexts $(P_{i_1},\ldots,P_{i_K})$ in candidate order.

\item \textbf{Concurrent pretrained encoding and BFV encryption.} While the Owner computes $C_K$, the User runs the unchanged pretrained encoder on the same tokens, producing $\mathbf{q}=\mathsf{Normalize}(E_{\mathrm{pre}}(x))$. The User quantizes $\mathbf{q}$ with the shared scale, and sends $\mathsf{Enc}(\bar{\mathbf{q}})$. Thus the latency before scoring is the hash encoder pass followed by $\max\{T_{\mathrm{Hamming}},T_{\mathrm{pre}}+T_{\mathrm{enc}}\}$, and payload transfer begins as soon as the Hamming branch produces $C_K$.

\item \textbf{Packed BFV scoring.} Once both $C_K$ and $\mathsf{Enc}(\bar{\mathbf{q}})$ are ready, the Owner gathers the pretrained quantized candidate matrix $\bar{\mathbf{Z}}_K$ and computes $\mathsf{Enc}(\mathbf{s})=\bar{\mathbf{Z}}_K\mathsf{Enc}(\bar{\mathbf{q}})$ using plaintext--ciphertext multiplication and rotation-based sum reduction while payload streaming continues. At polynomial degree 8192, the \emph{compact} layout collects up to 8192 scores in one ciphertext, and randomized evaluation~\cite{hwang2025ciphertext} clears non-score slots and makes the evaluated ciphertext simulatable from the prescribed scores while preserving exact integer outputs (\S\ref{sec:bfv-correctness}); it requires neither ciphertext--ciphertext multiplication nor bootstrapping.

\item \textbf{Encrypted-score return.} The Owner returns the encrypted scores in the same local order as $C_K$. Candidate position $j\in[K]$ is the only selection coordinate exposed to the User; Owner document identifiers are never transmitted.

\item \textbf{Decrypt and select top-$k$.} The User decrypts the exact integer score vector $\mathbf{s}$ and chooses local positions $c_1,\ldots,c_k\in[K]$.

\item \textbf{Active-secure $k$-out-of-$K$ key retrieval.} The parties execute a batch of $k$ active-secure 1-out-of-$K$ Orrù--Orsini--Scholl (OOS) OT-extension transfers~\cite{orru2016actively}, one for each $c_j$. The receiver obtains one 128-bit OT key per row, while the sender derives $K$ option keys per row and masks every candidate content key in a $k\times K$ table. The Owner sends this masked-key table after OT; row $j$ and column $c_j$ reveal $\kappa_{i_{c_j}}$ to the User, which opens the corresponding buffered payload. The OT structure limits a malicious receiver to at most one content key per row and hides each $c_j$ from the Owner.

\item \textbf{Payload opening and round commit.} After authenticating and decrypting the selected AEAD payloads, the User confirms completion and the Owner commits the quota reservation. Any protocol or malicious-check failure aborts the session and releases the uncommitted reservation.
\end{enumerate}

\subsection{Design Rationale}
\label{sec:rationale}

\textbf{Metric DP.} Hamming search requires the Owner to receive a coarse query code, which exposes the query's neighborhood and creates matching, property-inference, and linkability channels. Metric DP gives this released code a distance-calibrated indistinguishability guarantee while preserving efficient plaintext filtering.

\textbf{Single-release utility recovery.} MESS~\cite{cui2026mess} recovers recall after discrete randomized response through 64 HNSW shards with separately trained IsoHash mappings, routing each item to 16 shards and searching every shard for candidates. This raises a relevant item's recovery probability to $1-(1-P_{\mathrm{hit}})^{16}$; however, the hash codes appearing in 16 releases also compose privacy leakage across shards, leading to extremely large $\epsilon$ and almost null formal guarantee. Our learned filter, instead, releases one metric-DP code and recovers utility by enlarging $K$, leading to a good balance between privacy, efficiency, and search quality.

The Owner computes $C_K$ from the released code and its own index, so by post-processing Theorem~\ref{thm:owner-view} covers the full trace, corpus knowledge included. Table~\ref{tab:coarse-neighborhood-case} gives an operational view of this coarse release. The closest DP-Hamming codes match isolated cues such as London, poppies, or tower without identifying the requested fact. The answer-bearing passage appears only at Hamming rank 1,439; encrypted clean-query scoring promotes it to rank 5, inside the User's hidden $k=10$ selection. The large candidate set therefore preserves the answer while separating coarse code proximity from precise semantic relevance. Appendix~\ref{app:candidate-cases} presents three additional queries spanning factual counts, locations, and calendar rules.

\begin{table}[t]
\small
\centering
\caption{One NQ query (\texttt{test1054}) under E5 pure-vMF ($\varepsilon=64$, $K=3000$, $k=10$).}
\label{tab:coarse-neighborhood-case}
\begin{tabular}{@{}p{0.17\columnwidth}p{0.16\columnwidth}p{0.57\columnwidth}@{}}
\toprule
View & Rank / $d_H$ & Text excerpt \\
\midrule
Target & -- & \emph{Who made the poppies at Tower of London?} \\
DP-Ham. & 1 / 65 & \emph{10 Downing Street:} ``The terrace and garden were constructed in 1736~\ldots'' \\
DP-Ham. & 2 / 69 & \emph{Anzac Day:} ``Paper poppies are widely distributed~\ldots'' \\
DP-Ham. & 3 / 70 & \emph{Eiffel Tower:} ``26 December 1888: Construction of the upper stage.'' \\
Stage~2 & 1439$\rightarrow$5 / 88 & \emph{Blood Swept Lands and Seas of Red:} ``The artist was Paul Cummins, with setting by stage designer Tom Piper.'' \\
\bottomrule
\end{tabular}
\end{table}

\textbf{Low-depth HE.} The scoring stage is a plaintext--ciphertext matrix--vector multiplication over pre-normalized vectors. Our exact-integer BFV scoring uses SIMD multiplication and rotation-based reduction in a compact layout of multiplicative depth 2. Local decryption handles top-$k$, eliminating encrypted comparison, ciphertext--ciphertext multiplication, and bootstrapping. The compact layout also clears non-output slots and randomizes evaluation according to ciphertext-simulatable BFV~\cite{hwang2025ciphertext}: BFV hides the query from the Owner, while the randomized output reveals no Owner operand information to a conforming secret-key User beyond the $K$ prescribed scores.

\textbf{Key-only OT.} Returning all $K$ documents would disclose the full candidate payload set, whereas transferring full documents inside OT would make the active-secure OT payload proportional to document size. The protocol therefore sends fixed-size 128-bit content keys through $k$ 1-out-of-$K$ choices and delivers AEAD ciphertexts on the ordinary channel. This composition keeps OT small, hides the selected positions from the Owner, and caps payload decryption at $k$ choices; authenticated quotas separately govern the released score vector.

\textbf{Owner-local execution.} The Owner already holds the corpus and can retain the binary index on its own infrastructure. The index occupies $NL/8$ bytes (e.g., 283~MB for $N{=}8.84$M at $L{=}256$), allowing the Hamming scan and HE scoring to run without a query-processing intermediary.

\section{Security Contract and Protocol Guarantees}
\label{sec:dp}

We first fix each party's view (\S\ref{sec:security-contract}), then prove query privacy against the Owner (\S\ref{sec:directional-release}), exact and simulatable scoring for the User (\S\ref{sec:bfv-correctness}), and the score and payload bounds (\S\ref{sec:user-access}).

\subsection{Explicit Per-Party Views}
\label{sec:security-contract}

For one round, the Owner's explicit leakage is
\begin{equation}
\begin{aligned}
\mathcal{L}_O=(&\mathsf{id},\mathsf{session},\mathsf{round},K,k,\mathsf{layout},\mathsf{lengths},\\
&\mathsf{timing},\tilde b_q,C_K,\mathsf{accept/reject}).
\end{aligned}
\label{eq:owner-leakage}
\end{equation}
We call two executions on queries \(x,x'\) \emph{auxiliary-equivalent} if they agree on every field of Equation~\ref{eq:owner-leakage} except \((\tilde b_q,C_K,\mathsf{accept/reject})\). Of these, only \(\tilde b_q\) depends on the query; the Owner computes \(C_K\) and the quota decision from \(\tilde b_q\) and its own state. Theorem~\ref{thm:owner-view} also covers the BFV and OT transcripts in the Owner's real view.

The application-level disclosure to a conforming User is
\begin{equation}
\mathcal{L}^{\mathrm{target}}_U=(K,\mathbf{s},\{|P_i|:i\in C_K\},\{D_i:i\in S,\ |S|\le k\},\mathsf{linkage}),
\label{eq:user-leakage}
\end{equation}
where \(\mathbf{s}\) is the candidate-local score vector, \(P_i\) is a padded payload ciphertext, and \(\mathsf{linkage}\) records equality of recurring ciphertexts. Theorem~\ref{thm:he-user-view} shows that HE scoring reveals nothing beyond \(\mathbf{s}\) to a conforming User, and Theorem~\ref{thm:payload-access} shows that OT reveals at most \(k\) payloads even to a malicious one.

\subsection{Directional Candidate and Query Privacy}
\label{sec:directional-release}

For query \(x\) with hash logits \(g(x)\in\mathbb{R}^L\), let \(h(x)=\tanh(\beta g(x))\in[-1,1]^L\) and \(u(x)=h(x)/\|h(x)\|_2\in\mathbb{S}^{L-1}\). Given \(u\), the mechanism samples \(Y\sim\mathrm{vMF}(u,\kappa)\), with density \(p_u(y)=C_L(\kappa)\exp(\kappa u^\top y)\), and releases \(M_\kappa(u)=\sign(Y)\).

\begin{theorem}[Pure directional metric privacy]
\label{thm:vmf-pure}
For every \(\kappa\ge0\), every \(u,u'\in\mathbb{S}^{L-1}\), and every \(S\subseteq\{-1,1\}^L\),
\begin{equation}
\Pr[M_\kappa(u)\in S]\le \exp\!\left(\kappa d_c(u,u')\right)\Pr[M_\kappa(u')\in S].
\label{eq:vmf-metric}
\end{equation}
Consequently, for angular radius \(\rho\in(0,\pi]\), choosing
\begin{equation}
\kappa=\frac{\varepsilon}{2\sin(\rho/2)}
\label{eq:vmf-calibration}
\end{equation}
gives \((\varepsilon,0)\)-directional privacy at radius \(\rho\) under \(d_\theta\) (Definition~\ref{def:directional-privacy}).
\end{theorem}

\begin{proof}
The vMF normalizer \(C_L(\kappa)\) does not depend on the mean direction, so for every \(y\in\mathbb{S}^{L-1}\), Cauchy--Schwarz gives \(\log(p_u(y)/p_{u'}(y))=\kappa(u-u')^\top y\le\kappa\|u-u'\|_2=\kappa d_c(u,u')\). Integrating over \(\sign^{-1}(S)\) yields Equation~\ref{eq:vmf-metric}. If \(d_\theta(u,u')\le\rho\le\pi\), then \(d_c(u,u')=2\sin(d_\theta(u,u')/2)\le2\sin(\rho/2)\), and Equation~\ref{eq:vmf-calibration} makes the exponent at most \(\varepsilon\).
\end{proof}

Equation~\ref{eq:vmf-metric} is a metric-DP guarantee under chordal distance~\cite{chatzikokolakis2013broadening} that holds for every pair of directions; the radius only fixes the calibration.

\begin{theorem}[Computational directional privacy of the Owner view]
\label{thm:owner-view}
Suppose a conforming User generates both queries, BFV is indistinguishable under chosen-plaintext attack (IND-CPA) given the public and Galois keys (under the standard circular-security assumption), and the OT is receiver-private against a semi-honest sender. Let \(x,x'\) be queries with auxiliary-equivalent executions, \(u=u(x)\), \(u'=u(x')\), and \(\eta=\kappa d_c(u,u')\), where \(e^\eta\) is polynomially bounded in the security parameter \(\lambda\). Then for every probabilistic polynomial-time (PPT) distinguisher \(A\),
\begin{equation}
\Pr[A(\mathsf{View}_O(x))=1]\le e^\eta\Pr[A(\mathsf{View}_O(x'))=1]+\operatorname{negl}(\lambda).
\label{eq:owner-view}
\end{equation}
\end{theorem}

\begin{proof}[Proof sketch]
Let \(H(x)\) be the view in which the OT receiver messages are simulated (receiver privacy) and the query ciphertext encrypts zero (IND-CPA); then \(|\Pr[A(\mathsf{View}_O(x))=1]-\Pr[A(H(x))=1]|\le\operatorname{negl}(\lambda)\), and likewise for \(x'\). Apart from the fixed auxiliary fields, \(H(x)\) is a randomized function of \(M_\kappa(u)\) alone, so Theorem~\ref{thm:vmf-pure} and post-processing give \(\Pr[A(H(x))=1]\le e^\eta\Pr[A(H(x'))=1]\). Chaining the three steps leaves an additive error \((1+e^\eta)\operatorname{negl}(\lambda)\), which is negligible because \(e^\eta\) is polynomially bounded.
\end{proof}

For \(T\) adaptively chosen query pairs, sequential composition replaces \(\eta\) in Equation~\ref{eq:owner-view} by \(\eta_T=\kappa\sum_{t=1}^{T}d_c(u_t,u'_t)\), which is at most \(T\varepsilon\) when every pair lies within the radius, provided \(e^{\eta_T}\) stays polynomially bounded. Exact repeats reuse the User's cached code and add no loss, as in RAPPOR's memoized randomized response~\cite{erlingsson2014rappor}; a client-side budget caps the remaining fresh releases.

\subsection{Exact Scoring and HE Privacy Scope}
\label{sec:bfv-correctness}

Let the quantized query \(\bar{\mathbf q}\) and every candidate \(\bar{\mathbf z}_i\) lie in \([-B,B]^d\), packed as in Appendix~\ref{app:packed-invariants}, which assigns each candidate one output slot (its \emph{anchor}).

\begin{lemma}[Exact integer scoring]
\label{thm:bfv-exact}
If the evaluated ciphertext decrypts correctly (its noise stays within the decryption bound) and
\begin{equation}
t>2dB^2,
\label{eq:bfv-bound}
\end{equation}
the centered residue decoded at the anchor of every candidate \(i\in C_K\) equals \(s_i=\sum_{j=1}^{d}\bar z_{ij}\bar q_j\in\mathbb{Z}\).
\end{lemma}

\begin{proof}
Appendix~\ref{app:packed-invariants} shows that the compact layout places \(s_i\bmod t\) at the anchor of candidate \(i\). Since \(|s_i|\le dB^2<t/2\), the centered representative of \(s_i\bmod t\) is \(s_i\) itself.
\end{proof}

With \(B=127\) and \(d=768\), Equation~\ref{eq:bfv-bound} requires \(t>24{,}774{,}144\), which the deployed 25-bit batching prime satisfies; noise correctness is checked separately (Appendix~\ref{app:packed-invariants}).

BFV IND-CPA hides \(\bar{\mathbf q}\) from the Owner. To hide the Owner's plaintext operands from the decrypting User, the compact layout instantiates ciphertext-simulatable BFV~\cite{hwang2025ciphertext}. Write \(R_t=\mathbb Z_t[X]/(X^n+1)\), identify ring elements with coefficient vectors, and let \(D_{\Lambda,\sigma}\) denote the discrete Gaussian on lattice coset \(\Lambda\). To multiply by a plaintext \(\mu\in R_t\), the Owner samples a lift \(\widehat\mu\leftarrow D_{\mu+t\mathbb Z^n,\sigma}\) and evaluates
\begin{equation}
\begin{aligned}
\mathsf{PMultE}(ct,\mu)&=ct\cdot\widehat\mu+(e,0),\\
\mathsf{Rand}(pk)&=e_2pk+(e_0,e_1),
\end{aligned}
\label{eq:randomized-bfv}
\end{equation}
where \(e\leftarrow\lfloor D_{\mathbb R^n,\tau}\rceil\), \(e_1,e_2\leftarrow D_{\mathbb Z^n,\sigma_r}\), and \(e_0\leftarrow\lfloor D_{\mathbb R^n,\tau_r}\rceil\), with \((\sigma,\tau,\sigma_r,\tau_r)\) set as in~\cite{hwang2025ciphertext} (Appendix~\ref{app:protocol-implementation}). Each group ciphertext receives one \(\mathsf{PMultE}\) and one added \(\mathsf{Rand}(pk)\), followed by the public rotation, mask, and addition subcircuit; the mask \(p\in R_t\) leaves \(\mathbf{s}\) at the anchors and zero in every other slot. Finally, each output ciphertext receives one \(\mathsf{Rand}_f(pk)\), which equals \(\mathsf{Rand}(pk)\) except that \(e_0\leftarrow\lfloor D_{\mathbb R^n,s_f}\rceil\).

\begin{theorem}[HE server-input privacy for conforming Users]
\label{thm:he-user-view}
Let \(ct_q\) be a fresh symmetric BFV encryption of a bounded, canonically packed query, let \((\sigma,\tau,\sigma_r,\tau_r)\) satisfy the correctness and smoothing conditions of ciphertext-simulatable BFV~\cite{hwang2025ciphertext}, and let \(s_f\ge(nt/2)\,\eta_\epsilon(\mathbb Z^n)\) for a negligible \(\epsilon\). Under the corresponding ring-learning-with-errors (RLWE) assumption, the compact scoring view of a conforming secret-key User is computationally simulatable as
\begin{equation}
\mathsf{View}^{\mathrm{HE}}_U \approx_c \mathsf{Sim}(pk,sk,ct_q,\mathbf{s},\mathsf{metadata}),
\label{eq:he-user-simulation}
\end{equation}
where \(\mathbf{s}\) contains the prescribed \(K\) exact scores. Consequently, the evaluated ciphertext coefficients and all non-score slots reveal no information about the candidate embeddings beyond \(\mathbf{s}\) and public metadata.
\end{theorem}

\begin{proof}[Proof sketch]
After \(\mathsf{PMultE}\) and \(\mathsf{Rand}\), each group ciphertext has a pseudorandom mask component under RLWE, and its noise is a discrete Gaussian over a coset of \(\mathbb Z^n\) that encodes the group's slotwise products~\cite{hwang2025ciphertext}. The public subcircuit maps this noise by a fixed integer-linear map, plus key-switching noise determined by the mask components. Multiplication by \(p\) is invertible over \(\mathbb Z\), so this map alone preserves the coset and with it every pre-mask slot. The flood removes that dependence. Each rotated mask \(p'\) satisfies \(\eta_\epsilon(p'\mathbb Z^n)\le\|p'\|_1\eta_\epsilon(\mathbb Z^n)\le s_f\), so applying the convolution lemma~\cite{peikert2010efficient} once per group shows that the output noise is \(\epsilon\)-close to a discrete Gaussian over the coset of \(\mathbb Z^n\) fixed by the output plaintext \(\mathsf{Encode}(\mathbf{s},0,\ldots,0)\), with a covariance computable from \(ct_q\), \(sk\), and the public parameters. \(\mathsf{Rand}_f\) also re-randomizes the mask component, so a simulator holding \(sk\) samples the output from \((pk,sk,ct_q,\mathbf{s})\).
\end{proof}

Appendix~\ref{app:protocol-implementation} specifies the samplers and BFV parameters. The theorem assumes a fresh symmetric, canonical query ciphertext; malformed ciphertexts require a well-formedness proof. The scores \(\mathbf{s}\) remain explicit leakage.

\subsection{Exact-Score Exposure and Payload Access}
\label{sec:user-access}

Each accepted round releases the \(K\) exact scores \(\mathbf{s}\). Across rounds, let \(Q_i\in\mathbb{Z}^{r_i\times d}\) stack the \(r_i\) quantized queries that have scored document \(i\), so the User holds \(\mathbf{s}_i=Q_i\bar{\mathbf z}_i\), and let the Owner's ledger cap \(r_i\le R\) per account. Then
\begin{equation}
\begin{aligned}
\log_2|\operatorname{supp}(\mathbf{s})|&\le K\log_2(2dB^2+1),\\
r_i\le R&\Longrightarrow\operatorname{rank}(Q_i)\le\min\{R,d\},
\end{aligned}
\label{eq:score-accounting}
\end{equation}
The first bound caps the information in one round's scores and grows linearly with \(K\). The second holds for adaptively chosen well-formed queries: with \(R<d\), the linear system \(\mathbf{s}_i=Q_i\bar{\mathbf z}_i\) leaves \(\bar{\mathbf z}_i\) undetermined, although priors over embeddings can still support inference. The cap is enforced per authenticated account, since each Sybil identity would receive a fresh \(R\).

\begin{theorem}[Per-round payload-key bound]
\label{thm:payload-access}
Assume the OOS extension realizes active-secure 1-out-of-\(K\) OT for each of \(k\) rows~\cite{orru2016actively}, the hash that masks content keys is a random oracle, content keys are independent and uniform, and the payload cipher is an authenticated encryption scheme. Then for every PPT malicious receiver in one round there exist positions \(c_1,\ldots,c_k\in[K]\) such that its view is computationally indistinguishable from one in which every payload outside \(\{c_1,\ldots,c_k\}\) encrypts zeros of the same padded length. Hence it obtains at most \(k\) distinct payloads for any \(K\).
\end{theorem}

\begin{proof}[Proof sketch]
Active security lets a simulator extract one choice \(c_j\) per row and give the receiver only that row's option key. Every other option key is unknown to the receiver, so its random-oracle mask makes the corresponding content key uniform in the receiver's view; AEAD confidentiality under that key then permits replacing the payload by an encryption of zeros. A hybrid over the \(k\) rows and the unselected payloads proves the claim.
\end{proof}

Increasing \(K\) therefore enlarges the score disclosure in Equation~\ref{eq:score-accounting} but not the per-round payload cap. Across \(T\) accepted rounds the payload bound composes to \(Tk\), so quotas, not cryptography, limit eventual corpus enumeration. Scores, padded lengths, and linkage of recurring payload ciphertexts remain explicit leakage.

\section{Evaluation}
\label{sec:experiments}

Our evaluation answers three research questions. \RQ{1}~\textbf{Retrieval quality:} does the learned shortlist outperform classical and supervised hashing and match full-corpus retrieval across encoders, domains, and corpus scales, and does it keep retrieval and end-to-end RAG quality under formal DP? \RQ{2}~\textbf{Protocol efficiency:} how do candidate budget and pipeline overlap determine latency and communication, and how does the protocol compare with prior private-retrieval systems? \RQ{3}~\textbf{Representation privacy:} how much do binarization and DP randomization reduce embedding inversion and property inference?

\begin{table*}[thbp]
\caption{Two-stage retrieval. Pretr.\ marks full-corpus retrieval with the Stage~2 encoder; Recall@$K$ pairs the Hamming shortlist (bin.)\ with it; Top-10 kept is the share of its top-10 inside the Hamming top-$K$; $\Delta@K$ is S2@$K$ minus Pretr.\ Fl.}
\label{tab:main}
\centering
\small
\setlength{\tabcolsep}{5pt}
\begin{tabular}{@{}llrccccccccc@{}}
\toprule
 & & & \multicolumn{4}{c}{\textbf{Stage 1}} & \multicolumn{5}{c}{\textbf{Stage 2: NDCG@10}} \\
\cmidrule(lr){4-7}\cmidrule(lr){8-12}
 & & & \multicolumn{2}{c}{Recall@$K$ (bin./pretr.)} & \multicolumn{2}{c}{Top-10 kept (\%)} & & & & & \\
\cmidrule(lr){4-5}\cmidrule(lr){6-7}
\textbf{Model} & \textbf{Dataset} & \textbf{Docs} & \textbf{200} & \textbf{500} & \textbf{200} & \textbf{500} & \textbf{S2@200} & \textbf{S2@500} & \textbf{Pretr.\ Fl.} & $\boldsymbol{\Delta@200}$ & $\boldsymbol{\Delta@500}$ \\
\midrule
E5-base-v2 & SciDocs & 25K & .433/.495 & .547/.606 & 87.3 & 94.5 & .1875 & .1874 & .1870 & $+.0005$ & $+.0004$ \\
(256-bit) & NQ & 2.7M & .888/.968 & .927/.983 & 89.3 & 93.9 & .5723 & .5787 & .5854 & $-.0132$ & $-.0068$ \\
 & DBpedia-Entity & 4.6M & .455/.614 & .547/.706 & 77.8 & 85.6 & .4144 & .4224 & .4271 & $-.0127$ & $-.0047$ \\
 & Climate-FEVER & 5.4M & .530/.636 & .615/.721 & 62.2 & 71.1 & .2818 & .2785 & .2627 & $+.0192$ & $+.0158$ \\
 & FEVER & 5.4M & .937/.972 & .948/.977 & 75.8 & 82.9 & .8417 & .8451 & .8501 & $-.0084$ & $-.0050$ \\
\midrule
BGE-base & SciDocs & 25K & .521/.599 & .647/.710 & 95.6 & 98.9 & .2224 & .2225 & .2228 & $-.0004$ & $-.0003$ \\
(512-bit) & NQ & 2.7M & .892/.958 & .935/.978 & 91.7 & 95.4 & .5326 & .5372 & .5414 & $-.0088$ & $-.0042$ \\
 & DBpedia-Entity & 4.6M & .476/.604 & .568/.707 & 85.5 & 91.8 & .4018 & .4041 & .4081 & $-.0063$ & $-.0040$ \\
 & Climate-FEVER & 5.4M & .594/.685 & .677/.760 & 82.6 & 91.0 & .2874 & .2848 & .2836 & $+.0038$ & $+.0012$ \\
 & FEVER & 5.4M & .943/.973 & .952/.979 & 81.0 & 87.0 & .8480 & .8483 & .8495 & $-.0015$ & $-.0012$ \\
\bottomrule
\end{tabular}
\end{table*}

\subsection{Experimental Setup}

\textbf{Encoder, Datasets, and Metrics.} We train E5-base-v2~\cite{wang2022e5} and BGE-base-en-v1.5~\cite{chen2024bge} hash models on MS MARCO passage ranking~\cite{nguyen2016msmarco} and evaluate zero-shot transfer on five BEIR corpora~\cite{thakur2021beir}: SciDocs~\cite{cohan2020scidocs}, Natural Questions (NQ)~\cite{kwiatkowski2019natural}, DBpedia-Entity~\cite{hasibi2017dbpedia}, Climate-FEVER~\cite{diggelmann2020climate}, and FEVER~\cite{thorne2018fever}, following the standard BEIR zero-shot protocol. The corpora span from 25,657 to 5.4M documents. E5-base-v2 uses 256-bit codes, the setting of all protocol and privacy experiments, and BGE-base-en-v1.5 uses 512-bit codes to cover a second backbone at a larger width. Stage~1 reports Recall@$K$ over the relevance judgments for the learned Hamming filter; Stage~2 reranks exactly those candidates with the unchanged pretrained encoder and reports NDCG@10. Full-corpus retrieval with the same pretrained encoder is presented as a reference. Appendix~\ref{app:setup} gives the training recipe and hyperparameters; Appendix~\ref{app:objective-ablation} ablates each loss term and cross-evaluates both backbones at both widths.

\noindent\textbf{Representation Exposure.} Search-based inversion evaluates the embedding-guided search stages of ZSInvert~\cite{zhang2025universal} on 100 randomly sampled MS MARCO documents. Llama-3.1-8B-Instruct~\cite{llama2024herd} generates a width-50 beam for six search rounds, scored by float cosine or normalized hash similarity; we report mean verifier cosine and attack success at cosine 0.8. Generation-based inversion trains a GEIA~\cite{li2023geia} DialoGPT-medium~\cite{zhang2020dialogpt} decoder on PersonaChat~\cite{zhang2018personachat} for 10 epochs and reports token F1 and verifier cosine on held-out passages. Following the property-inference threat model of Song and Raghunathan~\cite{song2020information}, we evaluate AG News topic~\cite{zhang2015character}, IMDB sentiment~\cite{maas2011learning}, and 50-way 20 Newsgroups authorship~\cite{lang1995newsweeder} labels. Five stratified 60/20/20 splits separate attacker training, model selection, and victim testing; validation macro F1 selects the best model among logistic regression, a two-layer MLP, and LightGBM~\cite{ke2017lightgbm}, and we report victim-test macro F1 averaged across the five splits. Appendix~\ref{app:attacks} gives the complete attack flow, decoding limits, optimization hyperparameters, and preprocessing.

Following metric- and directional-DP evaluation conventions~\cite{chatzikokolakis2013broadening,weggenmann2021directional}, we report each operating point by its protected space, radius, and $(\varepsilon,\delta)$ parameters. Definition~\ref{def:directional-privacy} defines the generic radius $\rho$; here, we denote $\rho_h$ to be the Euclidean radius on the bounded pre-sign vector, and $\rho_\theta$ to be the angular radius on its normalized direction. The RDP-vMF comparison maps $\rho_h=2$ to $\rho_\theta=2\arcsin(1/\sqrt{L})$. The retrieval sweeps and property-inference table use $\varepsilon\in\{8,16,32,64\}$ at $\rho_h=2$; the inversion sweeps additionally include tighter budgets and the $\rho_h=6.32$ setting. Following standard approximate-DP calibration~\cite{dwork2014foundations}, Gaussian and RDP-vMF set $\delta=10^{-6}$, below the inverse of every evaluated query-set size, while pure-vMF provides $\delta=0$. The attack plots use the same calibrated mechanisms, and Theorem~\ref{thm:vmf-pure} gives the pure-vMF calibration.

\subsection{\texorpdfstring{\RQ{1}}{RQ1:}~Retrieval Quality}

\subsubsection{Plaintext Retrieval Quality}
\label{sec:search_quality}

\noindent\textbf{A concise shortlist preserves quality and may prune distractors.} Table~\ref{tab:main} evaluates the deployed two-stage pipeline at $K\in\{200,500\}$. Recall@$K$ measures how much judged-relevant material survives the learned filter; Stage~2 NDCG@10 measures the ranking obtained when the original pretrained model scores only those candidates. The full-corpus column uses the same scorer, and $\Delta@200$ and $\Delta@500$ isolate the candidate filtering capability of the hash model. Recall@$K$ counts judged-relevant documents, so Table~\ref{tab:main} pairs it with the full-corpus Recall@$K$ of the pretrained encoder and the share of that encoder's top-10 kept in the shortlist. On SciDocs, full-corpus E5 reaches Recall@500 .606 against .547 for our shortlist, yet both give NDCG@10 .187, because the shortlist keeps 94.5\% of the scorer's own top-10 (99.6\% at $K=3000$). \emph{Near-lossless} refers to this end-to-end quality ($\ge$99\% of full-corpus NDCG@10).

At $K=500$, both encoders retain 98.84--100.21\% of full-corpus NDCG@10 on SciDocs, NQ, DBpedia-Entity, and FEVER. On Climate-FEVER, the shortlist improves NDCG@10 by 0.0158 for E5 and 0.0012 for BGE. The learned filter removes high-scoring distractors: E5 keeps 71.1\% of the full-corpus top-10, and only 1.4\% of the dropped documents are judged relevant, against 11.1\% of those kept.

Raising $K$ from 200 to 500 adds 0.0064 and 0.0080 NDCG@10, respectively, and several encoder--dataset pairs already saturate at $K=200$. We therefore report both budgets here and benchmark latency through $K=3000$, leaving larger candidate pools available for the differential-privacy operating points evaluated next.

\noindent\textbf{The learned filter outperforms classical and supervised hashing baselines.} Table~\ref{tab:e5-hash-baselines} organizes candidate filters by the information used to construct their codes. Direct sign is a parameter-free, one-bit quantization of each pretrained coordinate. Data-independent LSH comprises random-hyperplane LSH~\cite{charikar2002similarity} and Super-Bit LSH~\cite{ji2012superbit}, which orthogonalizes the random projections. Unsupervised data-dependent hashing includes PCA-sign and the learned rotations of ITQ~\cite{gong2011itq} and IsoHash~\cite{kong2012isotropic}; we fit each transformation on MS MARCO passage embeddings and transfer it across corpora. BPR~\cite{yamada2021bpr} represents a recent supervised learning-to-hash method through a pairwise ranking objective, while our filter jointly adapts the encoder and hash head with retrieval and ranking supervision.

\begin{table}[t]
\centering
\caption{E5 candidate filters, averaged over the five corpora, with the same Stage~2 scorer. Direct sign uses 768 bits; the others use 256.}
\label{tab:e5-hash-baselines}
\resizebox{\columnwidth}{!}{%
\begin{tabular}{@{}lrrrr@{}}
\toprule
& \multicolumn{2}{c}{Recall@$K$} & \multicolumn{2}{c}{NDCG@10} \\
\cmidrule(lr){2-3}\cmidrule(lr){4-5}
Method & $K=200$ & $K=500$ & $K=200$ & $K=500$ \\
\midrule
Direct sign (768b) & .5046 & .5669 & .4183 & .4338 \\
Random-hyperplane LSH~\cite{charikar2002similarity} & .3068 & .3743 & .2893 & .3283 \\
Super-Bit LSH~\cite{ji2012superbit} & .3257 & .3916 & .3008 & .3366 \\
PCA-sign & .5802 & .6398 & .4469 & .4534 \\
ITQ~\cite{gong2011itq} & .6080 & .6789 & .4494 & .4562 \\
IsoHash~\cite{kong2012isotropic} & .6098 & .6786 & .4520 & .4573 \\
BPR$^{\dagger}$~\cite{yamada2021bpr} & .6125 & .6813 & .4473 & .4540 \\
\textbf{Learned filter (ours)} & \textbf{.6486} & \textbf{.7170} & \textbf{.4595} & \textbf{.4624} \\
\bottomrule
\end{tabular}%
}
\begin{minipage}{\columnwidth}
\footnotesize $^{\dagger}$For a fair comparison, we manually reimplement BPR using the same E5 representation, 256-bit budget, MS MARCO training data, symmetric Hamming candidate search, and evaluation pipeline as our method.
\end{minipage}
\end{table}

Table~\ref{tab:e5-hash-baselines} isolates candidate-filter quality by holding exact Hamming search and the pretrained Stage~2 scorer fixed. Our filter improves mean Recall@500 by 0.0357 and mean Stage~2 NDCG@10 by 0.0084 over BPR; it also improves mean Stage~2 NDCG@10 by 0.0051 over the strongest unsupervised baseline. The full-precision scorer can repair ordering only among documents retained by Stage~1, making candidate recall the more direct measure of hash quality.

\subsubsection{Retrieval under Differential Privacy}

Following the comparison methodology of Biswas et al.~\cite{biswas2025comparing}, we evaluate randomized response, analytic Gaussian, and RDP-vMF at common $(\varepsilon,\delta=10^{-6})$ targets and the protocol's formal pure-vMF mechanism under pure metric DP. Randomized response composes privacy on each bit across the full binary code; RDP-vMF calibrates its Rényi-divergence curve at the angular boundary defined above; Gaussian calibrates its analytic profile to Euclidean radius $\rho_h=2$ on the bounded pre-sign representation. Pure-vMF uses the exact directional calibration in Theorem~\ref{thm:vmf-pure}.

\DPMainUtilityTable

\noindent\textbf{Continuous randomization preserves retrieval quality.} For each continuous mechanism, Table~\ref{tab:dp-retrieval} lists the smallest candidate budget that recovers about 99\% of full-corpus NDCG@10, together with $K=3000$ at a tighter $\varepsilon$. With $\varepsilon=16$, RDP-vMF needs 1,000 candidates and Gaussian 2,000, both reaching 99.2\%; pure-vMF, whose $\varepsilon$ bounds pure metric DP, reaches 99.4\% with 2,000 candidates at $\varepsilon=64$. Randomized response keeps only 5.7\% even with 3,000 candidates. Scoring these larger shortlists adds 0.21--0.38 seconds to a 128-token Qwen3-32B~\cite{yang2025qwen3} RAG query (at most 6.5\%), and 0.57 seconds at $K=3000$. Appendix~\ref{app:dp-results} gives the per-dataset curves and the complete E5 and BGE sweeps.

\noindent\textbf{Tighter budgets with a larger shortlist.} Actual deployments can flexibly choose tighter budgets while retaining the quality by expanding candidate set. At $K=10{,}000$ (0.18\% of FEVER), pure-vMF at $\varepsilon=32$ retains 98\% of float NDCG@10 on FEVER and NQ, and Gaussian and RDP-vMF at $\varepsilon=8$ retain 94--97\%, at 3.9~seconds per query at 100~Mbps (Table~\ref{tab:eps-k-frontier}).

\subsubsection{End-to-End RAG Quality}

\noindent\textbf{Protected retrieval preserves end-to-end RAG quality.} We evaluate 500 Natural Questions queries over the full 2.68M-passage BEIR corpus~\cite{thakur2021beir}. The E5 scorer supplies the top five passages to Qwen3-32B, which returns a short answer under greedy decoding; exact match and token F1 use the NQ-Open answer aliases. Table~\ref{tab:rag-direct} shows that every two-stage operating point remains within 0.2 EM and 0.20 F1 of full-corpus float retrieval. Pure-vMF at $(\varepsilon=64,K=3000)$ reaches 50.0 EM and 62.89 F1, compared with 50.0 and 63.09 for the float reference. Appendix~\ref{app:rag} evaluates client-side cross-encoder reranking of the authorized payloads.

\begin{table}[thbp]
\centering
\caption{End-to-end RAG quality on 500 NQ queries. Hit is answer-alias coverage in 5 passages supplied to Qwen3-32B.}
\label{tab:rag-direct}
\small
\begin{tabular}{@{}lcrrr@{}}
\toprule
Retrieval & $(\varepsilon,K)$ & Hit & EM & F1 \\
\midrule
Full-corpus float & $(\infty,N)$ & 88.4 & 50.0 & 63.09 \\
No-DP hash & $(\infty,500)$ & 87.8 & 50.2 & 63.12 \\
Gaussian & $(16,3000)$ & 88.2 & 50.2 & 63.06 \\
RDP-vMF & $(16,3000)$ & 87.8 & 50.2 & 63.00 \\
Pure-vMF & $(64,3000)$ & 87.8 & 50.0 & 62.89 \\
\bottomrule
\end{tabular}
\end{table}

\subsection{\texorpdfstring{\RQ{2}}{RQ2:}~Protocol Efficiency}
\label{sec:efficiency}

We evaluate four sources of systems cost: the two-encoder pipeline, candidate-set scaling, complete RAG latency, and the online latency relative to prior private-retrieval systems. Appendices~\ref{app:protocol-implementation} and~\ref{app:latency-methodology} give the implementation and measurement details.

\noindent\textbf{Pipeline overlap absorbs most of the second encoder pass.} We measure the complete online path on SciDocs ($N=25{,}657$, $d=768$, $k=10$, $K=500$) over a 10-Gbps link, spanning shared tokenization, BF16 forward passes of the hash and Stage~2 encoders, randomized compact BFV scoring, active-secure OT, and final payload decryption. Model-stage results are means over 50 queries; cryptographic stages are medians of 20 rounds of deployed randomized evaluation with 128 Owner threads.
Figure~\ref{fig:latency} expands Steps~1--4 of the online protocol (\S\ref{sec:online}). After Step~1 sends the coarse code, Owner-side Hamming search (Step~2) and the User's pretrained encoder pass plus BFV encryption (Step~3) run concurrently. For E5, Step~2 finishes at 10.89~ms and starts payload streaming before Step~3 finishes at 13.26~ms, so Step~4 waits for Step~3. For BGE, Step~3 finishes at 13.34~ms before Step~2 fixes $C_K$ at 14.23~ms, so Step~4 instead waits for Step~2. The complete paths take 129.67 and 130.64~ms, respectively.

\begin{figure}[t]
\centering
\includegraphics[width=\linewidth]{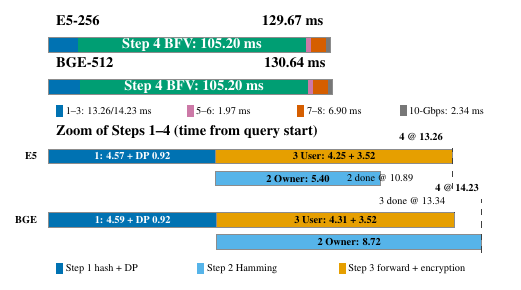}
\caption{Protected critical-path latency (ms) at $K=500$; \S\ref{sec:online} defines the steps.}
\label{fig:latency}
\end{figure}

\noindent\textbf{Candidate budget directly controls compute and communication.} We sweep $K$ from 200 to 3000 on SciDocs and measure the complete randomized-evaluation protocol between two processes connected by a loopback link that tc/netem leaves unshaped or limits to 1~Gbps or 100~Mbps (1~ms RTT). As Table~\ref{tab:e2e-k} shows, unshaped latency grows from 88~ms at $K=200$ to 581~ms at $K=3000$, while traffic grows from 1.64 to 13.65~MB. A 1-Gbps link adds no measurable latency because the payloads stream while HE scoring runs; at 100~Mbps, $K=500$ takes 297~ms and $K=3000$ takes 1.21~seconds, below the 1.67~seconds that a serial transfer would add (Appendix~\ref{app:latency-methodology}). Compact BFV scoring returns one score ciphertext throughout this range, so the remaining growth comes from candidate scoring and the payload and masked-key traffic, which scale with $K$. The smallest candidate budget that meets the retrieval-quality target therefore gives the best operating point.

\begin{table}[t]
\centering
\caption{Protected retrieval latency on SciDocs over tc/netem-shaped links (128 Owner threads, median of 20 rounds).}
\label{tab:e2e-k}
\small
\setlength{\tabcolsep}{3.2pt}
\begin{tabular}{@{}rrrrr@{}}
\toprule
& & \multicolumn{3}{c}{Protected latency (ms)} \\
\cmidrule(l){3-5}
$K$ & Traffic (MB) & Unshaped & 1~Gbps & 100~Mbps \\
\midrule
200 & 1.64 & 88 & 94 & 226 \\
500 & 2.93 & 120 & 129 & 297 \\
1000 & 5.07 & 209 & 206 & 569 \\
2000 & 9.36 & 389 & 434 & 831 \\
3000 & 13.65 & 581 & 565 & 1207 \\
\bottomrule
\end{tabular}
\end{table}

\noindent\textbf{Protection adds little latency to a full RAG pipeline.} The protected step costs 0.12--0.57~seconds for $K=500$--3000, against 9~ms for plaintext retrieval. We place the full 2.68M-passage E5 index on one A100 and serve Qwen3-32B with vLLM on a second A100, which generates 128 tokens in 5.744~seconds (median over 100 NQ queries). The plaintext filter and full-corpus float baselines take 5.753 and 5.750~seconds end to end. Figure~\ref{fig:rag-latency} reports the incremental protection cost: $K=500$ adds 0.121~seconds (2.1\%), the representative $K=1000$ and $K=2000$ operating points add 0.206 (3.6\%) and 0.375~seconds (6.5\%), and $K=3000$ adds 0.566~seconds (9.8\%). Appendix~\ref{app:latency-methodology} reports the measurement scope and output-length sensitivity.

\begin{figure}[t]
\centering
\includegraphics[width=\linewidth]{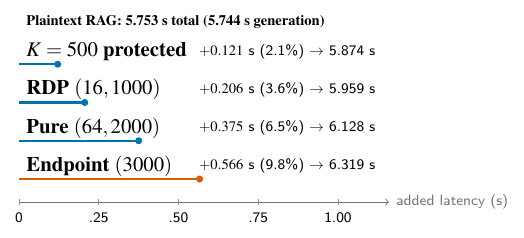}
\caption{Incremental latency over plaintext 128-token RAG (Qwen3-32B on vLLM, 5.753~s total).}
\label{fig:rag-latency}
\end{figure}

\noindent\textbf{At matched quality, our protocol is fastest at every evaluated corpus scale.} Figure~\ref{fig:protocol-comparison} compares online latency with P$^2$RAG~\cite{ming2026p2rag}, RemoteRAG~\cite{chang2025remoterag}, PANTHER~\cite{li2025panther}, and Pisces~\cite{liang2026pisces} on SciDocs (25.7K documents), Touch\'e (382.5K), and NQ-1M (1M), using the same E5-base-v2 embeddings, top-10 output, hardware, and 100-Mbps link. Our protocol uses pure-vMF at $\varepsilon=64$ with the smallest $K$ that retains 99\% of float NDCG@10 on each corpus (292, 104, and 1,956), and takes 0.264, 0.186, and 0.818~seconds; Touch\'e is fastest because it needs the fewest candidates. This is 2.4$\times$, 59.7$\times$, and 35.7$\times$ faster than P$^2$RAG and 11.5$\times$, 28.6$\times$, and 8.8$\times$ faster than RemoteRAG. PANTHER takes 19.40 and 40.41~seconds on SciDocs and Touch\'e and exhausts 256~GB of memory answering its NQ-1M PIR query. Pisces takes 70.0~seconds on SciDocs (265$\times$ ours) and 944~seconds on Touch\'e, and its first NQ-1M query does not finish within 27 minutes, because its filter passes 85--91\% of the documents to MPC scoring. Our BFV scorer uses 128 Owner threads, whereas P$^2$RAG uses 192 hardware threads and excludes its trusted-dealer preprocessing, both of which favor its online latency. Appendix~\ref{app:protocol-comparison} gives baseline configurations and 1-Gbps results.

\begin{figure}[t]
\centering
\includegraphics[width=\linewidth]{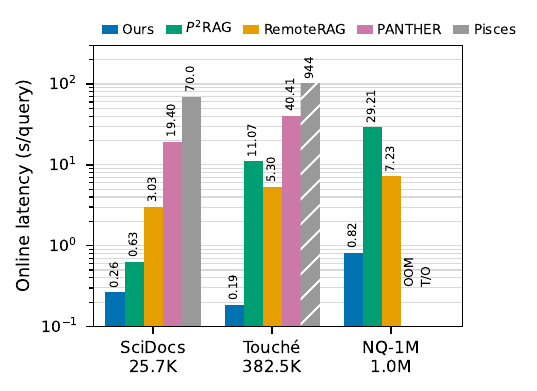}
\caption{Matched online latency at 100~Mbps (E5, top-10). The hatched Pisces bar is off-scale (944~s).}
\label{fig:protocol-comparison}
\end{figure}

\subsection{\texorpdfstring{\RQ{3}}{RQ3:}~Representation Privacy}
\label{sec:security_eval}

Deployed DP accepts large budgets (the US 2020 Census uses $\varepsilon=19.61$~\cite{census2021dp}), and reconstruction analyses show that large $\varepsilon$ still blocks reconstruction where the formal bound is loose~\cite{guo2023analyzing,balle2022reconstructing,hayes2023bounding}, so we measure the exposure directly.

\noindent\textbf{Binarization and DP provide defense in depth against inversion.} Table~\ref{tab:inversion-nodp} separates the two layers: the learned hash creates a discrete bottleneck that reduces leakage even without DP, and metric-DP randomization adds a formally calibrated second layer. Table~\ref{tab:eia-case-study} illustrates the aggregate trend on a common target. The float reconstruction recovers the named organization, its nonprofit status, and its activity; the learned hash turns the foundation-like acronym into a banking organization; and the randomized code produces an unrelated geographic passage. Appendix~\ref{app:inversion-cases} reports the complete decoded strings for this target and three additional cases.

\begin{table}[t]
\centering
\caption{Embedding inversion on E5 (256-bit codes); lower is better.}
\label{tab:inversion-nodp}
\small
\begin{tabular}{@{}lccc@{}}
\toprule
Metric & Float & No-DP hash & Pure-vMF ($\varepsilon=64$) \\
\midrule
Search cosine & .824 & .638 & .405 \\
Search success & .920 & .590 & .250 \\
Gen. token F1 & .545 & .368 & .237\\
Gen. cosine & .784 & .584 & .412 \\
\bottomrule
\end{tabular}
\par\medskip
\centering
\caption{One search-based inversion case (excerpts; Appendix~\ref{app:inversion-cases} gives full outputs).}
\label{tab:eia-case-study}
\footnotesize
\resizebox{\columnwidth}{!}{%
\begin{tabular}{@{}lp{0.59\columnwidth}r@{}}
\toprule
Release & Target or reconstruction excerpt & Cosine \\
\midrule
Target & ``Welcome to the U.S. High School Bowling Foundation ... promotes the growth of high school bowling.'' & -- \\
Float & ``US Bowling Foundation. The nonprofit organization ... 501(C) ...'' & .903 \\
Learned hash & ``US Bank is affiliated by National Bank Union ...'' & .464 \\
Gaussian, $\varepsilon=8$ & ``The country Australia encompasses expansive territories ...'' & $-.060$ \\
\bottomrule
\end{tabular}
}
\end{table}

\begin{figure}[t]
\centering
\includegraphics[width=\linewidth]{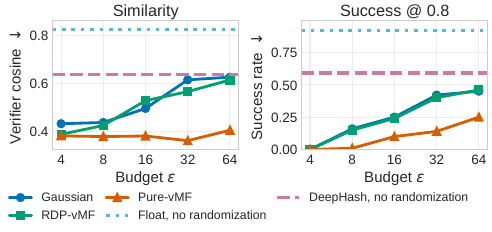}
\caption{Search-based embedding inversion under Gaussian, RDP-vMF, and pure-vMF randomization at $\rho_h=2$ ($\rho_\theta=0.125$ for pure-vMF). Success uses verifier cosine $\ge$0.8.}
\label{fig:eia-results}
\end{figure}

\begin{figure}[t]
\centering
\includegraphics[width=\linewidth]{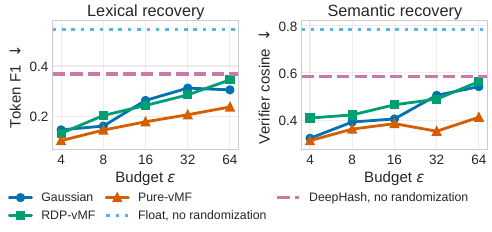}
\caption{Generation-based embedding inversion under Gaussian, RDP-vMF, and pure-vMF randomization.}
\label{fig:geia-results}
\end{figure}

\noindent\textbf{Tighter randomization suppresses embedding inversion.} Figure~\ref{fig:eia-results} shows that all three mechanisms suppress successful search-based reconstruction at tighter privacy budgets, degrading smoothly when $\varepsilon$ relaxes. Generation-based inversion exhibits the same privacy-budget response (Figure~\ref{fig:geia-results}): randomized releases reveal less lexical and semantic information than the learned code, with stronger suppression at smaller $\varepsilon$.

\begin{table}[!t]
\centering
\caption{E5 property-inference macro F1 for the validation-selected best attacker; lower is better. Entries average five stratified splits.}
\label{tab:dp-pia}
\resizebox{0.95\columnwidth}{!}{%
\begin{tabular}{@{}lrcccc@{}}
\toprule
Release & $\varepsilon$ & Topic & Sentiment & Authorship & Mean \\
\midrule
Gaussian & 8  & .6001 & .4763 & .0862 & .3875 \\
Gaussian & 16 & .6303 & .5402 & .1002 & .4235 \\
Gaussian & 32 & .7024 & .5849 & .1303 & .4725 \\
Gaussian & 64 & .8221 & .6575 & .1875 & .5557 \\
\addlinespace
RDP-vMF & 8  & .5980 & .5021 & .0727 & .3909 \\
RDP-vMF & 16 & .6251 & .5528 & .0924 & .4234 \\
RDP-vMF & 32 & .6713 & .6057 & .1235 & .4668 \\
RDP-vMF & 64 & .8325 & .6515 & .1550 & .5463 \\
\addlinespace
Pure-vMF & 8  & .4958 & .5813 & .0243 & .3671 \\
Pure-vMF & 16 & .6742 & .6156 & .0558 & .4485 \\
Pure-vMF & 32 & .7934 & .6664 & .1128 & .5242 \\
Pure-vMF & 64 & .8531 & .7143 & .1744 & .5806 \\
\midrule
Learned hash & $\infty$ & .8774 & .7383 & .2240 & .6133 \\
Float reference & $\infty$ & .8872 & .7986 & .3876 & .6911 \\
\bottomrule
\end{tabular}%
}
\end{table}

\noindent\textbf{Randomization further weakens property inference.} Property inference confirms the same layered effect (Table~\ref{tab:dp-pia}): learned hashing removes attribute signal relative to the float representation, and randomization further weakens topic, sentiment, and authorship inference as $\varepsilon$ tightens. Gaussian and RDP-vMF yield similar attack leakage and utility trade-offs. Pure-vMF, whose budget is pure DP ($\delta=0$), gives the lowest attack F1 at every $\varepsilon$ (mean .367--.581).

\section{Concluding Remarks}

We presented a practical two-party private dense-retrieval design for provider-held corpora serving external users. Concentrating private computation on a high-recall shortlist preserves retrieval quality and practical latency at million-document scale, while layered protections limit query and selection leakage and enforce the per-query payload allowance. 
Evaluations against embedding-inversion and property-inference attacks show that the released representation reduces reconstruction fidelity and attribute leakage. 
The resulting deployment reconciles interests that usually conflict: legitimate users receive privacy-preserving, precise semantic search, while corpus owners retain control over valuable content and align disclosure with the service's billing model.

\section*{Acknowledgement}
This work was in part supported by the Open Research Fund of State Key Laboratory of Internet Architecture, Tsinghua University, Beijing 100084, China.

\bibliographystyle{plain}
\bibliography{references}
\appendix

\section*{Ethical Considerations}
\label{sec:ethics}

This work proposes a defense solution in a privacy-sensitive two-party scenario, which aims to \textit{reduce} privacy exposure in dense retrieval against attacks documented in the literature~\cite{song2020information,morris2023text,li2023geia,huang2024transferable,chen2024text,zhang2025universal}. In particular, \textit{no} new attack algorithms beyond the reproduction and adaptation of existing work are presented.

\paragraph{Dual-use of attack reproduction.} We re-implement published embedding-inversion attacks~\cite{zhang2025universal,song2020information,li2023geia} to measure the exposure of the representations without protection and confirm the concerning leakage among them. The experiments reproduce existing attack capabilities on public benchmarks and pretrained encoders; they use \textit{no} proprietary corpus or production system.

\paragraph{Use of human-generated data.} All datasets are publicly available research benchmarks released for academic use~\cite{thakur2021beir,zhang2018personachat,lang1995newsweeder,zhang2015character,maas2011learning}. The study collects \textit{no} new data and involves \textit{no} interaction with human subjects.

\paragraph{Security boundary.} The formal query guarantee covers an honest-but-curious Owner observing the quantized randomized code, candidate set, and protocol metadata. For a conforming secret-key User, compact ciphertext-simulatable BFV restricts the scoring ciphertext to the explicit $K$-score oracle, while active-secure OT limits a malicious receiver's per-query payload-key recovery.
The released code is a research artifact that implements the declared security guarantees. A commercial deployment must additionally bind identities, protect endpoint keys, authenticate and encrypt both transport channels, and address side channels beyond the measured message lengths and timing.

\paragraph{Responsible disclosure.} The evaluation reproduces public attacks and addresses the general bi-encoder retrieval pipeline rather than a product-specific vulnerability, so coordinated vendor disclosure is not applicable.

\section{Proof Details and Auxiliary Calibration}
\label{app:privacy-proofs}

\subsection{Analytic Gaussian Baseline}

The additive Gaussian baseline acts on the bounded, unnormalized hash-head vector \(h\in[-1,1]^L\) rather than on the direction protected by Theorem~\ref{thm:vmf-pure}. For \(G_\sigma(h)=h+\mathcal{N}(0,\sigma^2I)\), two inputs at Euclidean distance \(D>0\) have the exact privacy profile~\cite{balle2018gaussian}
\begin{equation}
\delta_G(\varepsilon,D,\sigma)=
\Phi\!\left(\frac{D}{2\sigma}-\frac{\varepsilon\sigma}{D}\right)
-e^\varepsilon
\Phi\!\left(-\frac{D}{2\sigma}-\frac{\varepsilon\sigma}{D}\right),
\label{eq:analytic-gaussian}
\end{equation}
with \(\delta_G=0\) at \(D=0\). The profile increases with \(D\), so calibration at Euclidean radius \(\rho_h\) solves \(\delta_G(\varepsilon,\rho_h,\sigma)\le\delta_0\) at the boundary.

Table~\ref{tab:gaussian-calibration} gives the numerically solved scales used by our implementation for \(\rho_h=2\) and \(\delta_0=10^{-6}\). Each regression test substitutes the result into Equation~\ref{eq:analytic-gaussian} and checks the target profile, and the retrieval and attack sweeps use these calibrated scales.

\begin{table}[t]
\centering
\caption{Analytic Gaussian scales for bounded pre-sign representations at Euclidean radius \(\rho_h=2\) and \(\delta_0=10^{-6}\).}
\label{tab:gaussian-calibration}
\small
\begin{tabular}{@{}rrrrr@{}}
\toprule
\(\varepsilon\) & 8 & 16 & 32 & 64 \\
\midrule
\(\sigma\) & 1.3059 & 0.7372 & 0.4324 & 0.2638 \\
\bottomrule
\end{tabular}
\end{table}

Gaussian and vMF results protect different input spaces: \(G_\sigma\) uses Euclidean adjacency on \(h\), whereas \(M_\kappa\) uses angular or chord adjacency on \(u=h/\|h\|_2\). Their numerical privacy budgets must therefore be reported with the protected space and radius, not as a mechanism ranking under an unspecified common adjacency.

\subsection{Approximate-vMF Calibration}

For mean directions separated by angle \(\theta\), rotate coordinates so that \(u=e_1\) and \(u'=\cos\theta\,e_1+\sin\theta\,e_2\). Under \(Y\sim\mathrm{vMF}(u,\kappa)\), the privacy-loss random variable is
\begin{equation}
L(Y)=\kappa\bigl((1-\cos\theta)Y_1-\sin\theta Y_2\bigr).
\end{equation}
The corresponding pre-binarization hockey-stick divergence is
\begin{equation}
\label{eq:vmf-hockey-stick}
\begin{aligned}
\delta_{\mathrm{vMF}}(\varepsilon,\theta,\kappa)
&= \Pr_u\bigl[L(Y)>\varepsilon\bigr]\\
&\quad- e^\varepsilon\,\Pr_u\bigl[L(Y)<-\varepsilon\bigr].
\end{aligned}
\end{equation}
A radius-\(\rho\) approximate guarantee requires the supremum of Equation~\ref{eq:vmf-hockey-stick} over every \(\theta\in[0,\rho]\). Our quadrature evaluates the boundary, while certified approximate calibration additionally requires establishing the maximizing angle. The protocol uses the exact pure calibration of Theorem~\ref{thm:vmf-pure}; the RDP-vMF sweep serves as the empirical mechanism comparison.

\subsection{Packed-Layout Invariants}
\label{app:packed-invariants}

At degree 8192, BFV batching provides two rows of 4096 slots. Each row is divided into four 1024-slot segments, so a score group contains eight candidates. Query encryption repeats the \(d\)-coordinate vector in every segment and fills the remaining slots with zero; Owner encoding places one candidate in each corresponding segment.

After componentwise plaintext--ciphertext multiplication, rotations by \(1,2,\ldots,512\) and additions place each segment's dot product at its anchor. The compact layout then multiplies by a plaintext mask that retains only the eight anchors, rotates group \(g\) to residue \(g\bmod1024\), and adds 1024 groups per output ciphertext. Different groups then occupy different residues at every segment anchor, yielding \(\lceil K/8192\rceil\) ciphertexts. The final partial group and partial compact output contain only zero-padded dummy candidates.

These layout invariants establish where each modular score is decoded; Equation~\ref{eq:bfv-bound} establishes that its centered residue is the intended signed integer. Circuit correctness additionally assumes sufficient BFV noise budget, which the implementation checks empirically through successful decryption and oracle equality. Circuit privacy and malicious-input validity are separate properties governed by the scope in \S\ref{sec:bfv-correctness}.

\subsection{Cryptographic Hybrid Details}

For Theorem~\ref{thm:owner-view}, begin with the real Owner view for a conforming User. Receiver privacy of active OOS OT replaces the receiver-choice-dependent messages with a simulated transcript for the same sender inputs. BFV IND-CPA replaces the valid encryption of the clean int8 query by an encryption of zero; applying the public scoring circuit and serializing its output cannot increase the Owner's distinguishing advantage. The remaining query-dependent plaintext is \(M_\kappa(u)\) and its Hamming-search and quota post-processing. Directional privacy gives Equation~\ref{eq:owner-view}, while the auxiliary fields in Equation~\ref{eq:owner-leakage} are held fixed. Theorem~\ref{thm:owner-view} therefore governs the Owner view; \S\ref{sec:bfv-correctness} separately specifies the secret-key User view.

For Theorem~\ref{thm:payload-access}, active receiver security restricts each row to one OT option key. In the random-oracle hybrid, the hash of every unchosen option key is independent of the receiver's view, making its masked content key uniform. Replacing the unselected derived payload keys by random keys then reduces disclosure of any remaining plaintext to the confidentiality of ChaCha20--Poly1305; ciphertext modification is rejected by its authentication check. Across \(k\) rows, the receiver obtains at most \(k\) content keys, with duplicate selections yielding fewer distinct payloads.

\section{Experimental Details}
\label{app:ablation}

\subsection{Experimental Setup Details}
\label{app:setup}

We provide the training and evaluation details needed to reproduce the two-stage retrieval results in Table~\ref{tab:main}.

\textbf{Evaluation Corpora.} The five BEIR corpora originate from SciDocs~\cite{cohan2020scidocs}, Natural Questions~\cite{kwiatkowski2019natural}, DBpedia-Entity~\cite{hasibi2017dbpedia}, FEVER~\cite{thorne2018fever}, and Climate-FEVER~\cite{diggelmann2020climate}. Climate-FEVER adapts FEVER's claim--evidence methodology from artificially constructed general-domain claims to real-world climate claims and includes disputed evidence. Their dataset corpus are largely overlapping, but they have distinct query distributions.

\textbf{Metrics.} NDCG@$k$ is the standard normalized discounted cumulative gain at rank $k$, which rewards relevant documents appearing higher in the ranked list and normalizes against the ideal ordering. Recall@$k$ is the fraction of gold-relevant documents appearing in the top-$k$ results. In particular, the former is sensitive to ranking among the top-$k$, while the latter is only sensitive to whether the relevant documents appear, regardless of their specific ranking.

\textbf{Models and Hard Negatives.} E5-base-v2 uses mean pooling and a 256-bit hash head, whereas BGE-base-en-v1.5 uses its CLS representation and a 512-bit head. We retrieve 512 BM25~\cite{robertson1994bm25} candidates per MS MARCO training query, rerank them with the ``ms-marco-MiniLM-L6-v2''\footnote{https://huggingface.co/cross-encoder/ms-marco-MiniLM-L6-v2} cross-encoder~\cite{wang2020minilm}, retain the top 32, form a pool from the five highest-ranked non-relevant passages, and sample three negatives per query.
A mined negative is removed from the binary ranking loss when its cross-encoder score is within 0.5 of the labeled positive.
We also tried stronger cross-encoders, larger mining pools, or more hard-negatives per query during our experiments, but none of them gives noticeably better results, while more hard-negatives even harm the evaluated quality.

\textbf{Training Objective.} A batch contains queries $q_i$, labeled passages $p_i$, and mined-negative sets $\mathcal N_i$ for $i\in\{1,\ldots,B\}$. Let $\mathbf a_x$ and $\mathbf a_x^0$ denote the normalized pooled representations of text $x$ from the adapted and frozen encoders, respectively. The hash head produces logits $\mathbf z_x=W\,\mathrm{pool}(E_{\mathrm{LoRA}}(x))$ and training code $\mathbf h_x=\tanh(\beta\mathbf z_x)$; deployment uses $\bcode_x=\sign(\mathbf z_x)$, where $\sign(t)=+1$ for $t\ge 0$ and $-1$ otherwise. We write
\begin{equation}
\begin{aligned}
s^a(q,d)&=\mathbf a_q^\top\mathbf a_d,\\
s^0(q,d)&=(\mathbf a_q^0)^\top\mathbf a_d^0,\\
s^h(q,d)&=\mathbf h_q^\top\mathbf h_d/L.
\end{aligned}
\label{eq:training-scores}
\end{equation}
and let $c(q,d)$ be the cross-encoder relevance logit. For a score function $s$, candidate set $\mathcal D$, and temperature $\tau$, define
\begin{equation}
P_{i,\tau}^{s,\mathcal D}(d)=\frac{\exp(s(q_i,d)/\tau)}{\sum_{d'\in\mathcal D}\exp(s(q_i,d')/\tau)}.
\label{eq:training-soft-ranking}
\end{equation}
Let $\mathcal D_B=\{p_j\}_{j=1}^{B}\cup\bigcup_j\mathcal N_j$ be all passages in the batch. The implemented groups in Equation~\ref{eq:hash-objective-groups} expand as
\begin{equation}
\begin{aligned}
\mathcal{L}_{\mathrm{retrieval}} &= \mathcal{L}_{\mathrm{InfoNCE}}+\lambda_{\mathrm{bin}}\mathcal{L}_{\mathrm{bin}},\\
\mathcal{L}_{\mathrm{transfer}} &= \lambda_{\mathrm{rankKD}}\mathcal{L}_{\mathrm{rankKD}},\\
\mathcal{L}_{\mathrm{regularization}} &= \lambda_{\mathrm{floatKD}}\mathcal{L}_{\mathrm{floatKD}}+\lambda_{\mathrm{GOR}}\mathcal{L}_{\mathrm{GOR}}.
\end{aligned}
\label{eq:hash-objective-full}
\end{equation}
The continuous retrieval term is
\begin{equation}
\mathcal L_{\mathrm{InfoNCE}}=-\frac{1}{B}\sum_{i=1}^{B}\log P_{i,\tau}^{s^a,\mathcal D_B}(p_i).
\label{eq:infonce-loss}
\end{equation}
For false-negative margin $m$, we retain $\mathcal M_i=\{n\in\mathcal N_i:c(q_i,n)<c(q_i,p_i)-m\}$ and define $\mathcal I=\{i:\mathcal M_i\ne\varnothing\}$. The ranking term applied to the deployed binary codes is
\begin{equation}
\begin{aligned}
\ell_i&=\tau_b\log\sum_{n\in\mathcal M_i}e^{(s^h(q_i,n)-s^h(q_i,p_i))/\tau_b},\\
\mathcal{L}_{\mathrm{bin}}&=\frac{1}{|\mathcal I|}\sum_{i\in\mathcal I}\operatorname{softplus}(\ell_i/\tau_b).
\end{aligned}
\label{eq:binary-ranking-loss}
\end{equation}
Here $\operatorname{softplus}(u)=\log(1+e^u)$. E5 uses this direct Hamming-space term; BGE obtains stronger candidate recall from RankKD alone and sets $\lambda_{\mathrm{bin}}=0$.

RankKD (ranking knowledge distillation) transfers the adapted continuous ranking over all in-batch passages:
\begin{equation}
\mathcal L_{\mathrm{rankKD}}=\frac{1}{B}\sum_{i=1}^{B}\mathrm{KL}\!\left(P_{i,\tau}^{s^a,\mathcal D_B}\,\middle\|\,P_{i,\tau}^{s^h,\mathcal D_B}\right).
\label{eq:rankkd-loss}
\end{equation}
RankKD transfers batch-wide ordering without using labels from an evaluation corpus; the cross-encoder contributes difficult local examples through hard-negative mining.

The geometry terms stabilize the representation that supplies these rankings. For the batch collections $\mathcal X_Q=\{q_i\}_i$, $\mathcal X_P=\{p_i\}_i$, and $\mathcal X_N=\biguplus_i\mathcal N_i$, FloatKD (float-embedding distillation) aligns the adapted and frozen encoders:
\begin{equation}
\mathcal L_{\mathrm{floatKD}}=\frac{1}{3}\sum_{G\in\{Q,P,N\}}\frac{1}{|\mathcal X_G|}\sum_{x\in\mathcal X_G}\left[1-\cos\!\left(\mathbf a_x,\mathbf a_x^0\right)\right].
\label{eq:floatkd-loss}
\end{equation}
For any batch collection $\mathcal X$, GOR (global orthogonal regularization) uses the spread-out penalty
\begin{equation}
R_{\mathrm{GOR}}(\mathcal X)=\frac{1}{|\mathcal X|(|\mathcal X|-1)}\sum_{\substack{x,y\in\mathcal X\\x\ne y}}\cos^2(\mathbf h_x,\mathbf h_y),
\label{eq:gor-set-loss}
\end{equation}
We use $\mathcal L_{\mathrm{GOR}}=[R_{\mathrm{GOR}}(\mathcal X_Q)+R_{\mathrm{GOR}}(\mathcal X_P)+R_{\mathrm{GOR}}(\mathcal X_N)]/3$. FloatKD retains the pretrained geometry during LoRA adaptation, whereas GOR discourages different batch examples from collapsing to similar directions before binarization.

\textbf{Discrete Optimization and Hyperparameters.} The reported models use the differentiable code $\mathbf h=\tanh(\beta\mathbf z)$ throughout training; $\beta$ rises linearly from 1 to 2.5 during the first quarter and then remains fixed, and deployment applies $\sign(\mathbf z)$. We apply LoRA to every encoder linear layer with rank 16, $\alpha=16$, and dropout 0.05. Each epoch contains 300 steps at batch size 128, and both models train for 16 epochs. The encoder learning rate is $2\times10^{-6}$ for E5 and $5\times10^{-6}$ for BGE, while the hash-head rate is $2\times10^{-4}$ for both. InfoNCE has unit weight; $\lambda_{\mathrm{bin}}=0.8$ for E5 and 0 for BGE. We set $\lambda_{\mathrm{rankKD}}=\lambda_{\mathrm{floatKD}}=\lambda_{\mathrm{GOR}}=1$, with $\tau=0.05$, $\tau_b=0.1$, and $m=0.5$. Evaluation uses an exponential moving average with decay 0.999. Training one hash model takes 11.6~hours on one A100, including 2.4~hours of hard-negative mining; the corpus is encoded once offline.

\textbf{Two-Forward Evaluation.} Stage~1 runs the trained hash encoder to retrieve exact Hamming top-$K$ candidates. Stage~2 independently runs the original pretrained encoder and reranks only those candidates by float similarity; it never uses the hash model's continuous representation. Both encoder passes share tokenization and host-to-device transfer. E5 uses the standard ``query:'' and ``passage:'' prefixes, BGE uses its standard query instruction and no document prefix, and maximum query and document lengths are 48 and 512 tokens. We evaluate in BF16 with the LoRA weights merged into the Stage~1 encoder.

\textbf{DP Retrieval Sweeps.} We evaluate E5 and BGE on SciDocs, NQ, and FEVER at $K\in\{200,500,1000,2000,3000\}$. Gaussian, RDP-vMF, and pure-vMF use $\varepsilon\in\{8,16,32,64\}$; Gaussian and RDP-vMF set $\delta=10^{-6}$, while pure-vMF uses the exact calibration in Theorem~\ref{thm:vmf-pure}. Each randomized query code is searched by exact Hamming distance, and the unchanged pretrained encoder reranks the resulting candidates.

\textbf{Hash Baselines.} Table~\ref{tab:e5-hash-baselines} uses the same cached pretrained E5 query and corpus embeddings and the same exact Hamming and Stage~2 routines as Table~\ref{tab:main}. Direct sign thresholds all 768 embedding coordinates at zero. Random-hyperplane LSH~\cite{charikar2002similarity} uses a fixed 256-column Gaussian projection; Super-Bit LSH~\cite{ji2012superbit} orthogonalizes all 256 columns as one maximum-depth Super-Bit. PCA-sign, ITQ~\cite{gong2011itq}, and IsoHash~\cite{kong2012isotropic} are fitted on normalized embeddings of the first 20,000 MS MARCO corpus passages, with 50 ITQ and 100 IsoHash rotation updates, then applied without target-corpus fitting. Every Stage~2 score uses the unchanged full-precision pretrained E5 embedding, so the reported differences come only from candidate membership in Stage~1.

\subsection{Objective Ablation and Code Widths}
\label{app:objective-ablation}

Table~\ref{tab:objective-ablation} removes one loss term at a time from the E5 recipe and reports Stage~1 Recall@500 on the five corpora of Table~\ref{tab:main}. InfoNCE is the base retrieval loss and is kept in every variant. Removing RankKD costs the most (.019 mean Recall@500), followed by $\mathcal L_{\mathrm{bin}}$ (.010), FloatKD (.006), and GOR (.003).

\begin{table}[H]
\centering
\caption{E5 (256-bit) objective ablation: Stage~1 Recall@500.}
\label{tab:objective-ablation}
\footnotesize
\setlength{\tabcolsep}{4pt}
\begin{tabular}{@{}lrrrrrr@{}}
\toprule
Variant & SciDocs & NQ & DBpedia & C-FEVER & FEVER & Mean \\
\midrule
Full recipe & .547 & .927 & .547 & .615 & .948 & .717 \\
$-$RankKD & .537 & .904 & .514 & .606 & .931 & .698 \\
$-\mathcal L_{\mathrm{bin}}$ & .544 & .933 & .530 & .586 & .943 & .707 \\
$-$FloatKD & .551 & .920 & .517 & .623 & .943 & .711 \\
$-$GOR & .548 & .924 & .535 & .614 & .947 & .714 \\
\bottomrule
\end{tabular}
\end{table}

Table~\ref{tab:bits-backbone} applies the same recipe to both backbones at both code widths. Each backbone gains .026--.034 mean Recall@500 from 256 to 512 bits, and the two backbones stay within .013 of each other at either width. BGE sets $\lambda_{\mathrm{bin}}=0$ because RankKD alone gave slightly higher recall (.756 versus .751 with $\lambda_{\mathrm{bin}}=0.8$); pooling follows each backbone (mean pooling for E5, CLS for BGE).

\begin{table}[H]
\centering
\caption{Mean Stage~1 Recall@500 over the five corpora by backbone and code width.}
\label{tab:bits-backbone}
\small
\begin{tabular}{@{}lrr@{}}
\toprule
Backbone & 256 bits & 512 bits \\
\midrule
E5-base-v2 & .717 & .743 \\
BGE-base-en-v1.5 & .722 & .756 \\
\bottomrule
\end{tabular}
\end{table}

\subsection{Qualitative Candidate-Set Cases}
\label{app:candidate-cases}

Table~\ref{tab:candidate-cases-extra} extends Table~\ref{tab:coarse-neighborhood-case} with three NQ queries under the same E5 pure-vMF operating point. The nearest Hamming items preserve broad cues---constitutional amendments, elevation and battlefields, or quarters of a year---but do not supply the requested count, location, or calendar rule. The answer-bearing passages occur much deeper in the coarse ranking and move into the User's top ten only after clean-query Stage~2 scoring.

\begin{table*}[t]
\centering
\caption{Three additional NQ candidate-set cases under E5 pure-vMF ($\varepsilon=64$, $K=3000$, $k=10$, 256 bits). Stage~2 entries show coarse Hamming rank $\rightarrow$ clean-float rank.}
\label{tab:candidate-cases-extra}
\footnotesize
\setlength{\tabcolsep}{3pt}
\begin{tabular}{@{}cl>{\centering\arraybackslash}p{0.10\textwidth}>{\raggedright\arraybackslash}p{0.68\textwidth}@{}}
\toprule
Case & View & Rank / $d_H$ & Target or passage excerpt \\
\midrule
A & Target & -- & \emph{How many amendments to the Constitution have there been?} \\
& DP-Ham. & 1 / 60 & \emph{First Amendment to the United States Constitution:} ``The civil rights of none shall be abridged~\ldots'' \\
& DP-Ham. & 2 / 60 & \emph{Limited government:} ``The Ninth and Tenth Amendments~\ldots'' \\
& DP-Ham. & 3 / 60 & \emph{Article Three of the United States Constitution:} ``Hamilton continues~\ldots'' \\
& Stage~2 & 1742$\rightarrow$1 / 82 & \emph{List of amendments to the United States Constitution:} ``Thirty-three amendments~\ldots{} have been proposed~\ldots{} Twenty-seven~\ldots{} are part of the Constitution.'' \\
\addlinespace
B & Target & -- & \emph{Where is the world's highest battlefield located?} \\
& DP-Ham. & 1 / 68 & \emph{Operation Market Garden:} ``The country was wooded and rather marshy~\ldots{} two important hills~\ldots{} represented some of the highest ground in the Netherlands.'' \\
& DP-Ham. & 2 / 70 & \emph{List of elevation extremes by country:} ``The Dead Sea is the lowest point on Earth.'' \\
& DP-Ham. & 3 / 70 & \emph{Geography of China:} ``Tallest mountain peaks.'' \\
& Stage~2 & 2317$\rightarrow$2 / 87 & \emph{Siachen Glacier:} ``The glacier's region is the highest battleground on Earth~\ldots{} Pakistan and India~\ldots{} maintain a permanent military presence.'' \\
\addlinespace
C & Target & -- & \emph{Explain what happens to the extra quarter of a day each calendar year.} \\
& DP-Ham. & 1 / 67 & \emph{Calendar year:} ``The calendar year can be divided into four quarters~\ldots'' \\
& DP-Ham. & 2 / 69 & \emph{Fiscal year:} ``The Financial year is split into the following four quarters.'' \\
& DP-Ham. & 3 / 71 & \emph{Accounting period:} ``The end of the fiscal year would move one day earlier~\ldots'' \\
& Stage~2 & 555$\rightarrow$8 / 88 & \emph{Leap year:} ``Adding one extra day in the calendar every four years compensates for~\ldots{} almost 6 hours.'' \\
\bottomrule
\end{tabular}
\end{table*}

Together with Table~\ref{tab:coarse-neighborhood-case}, these cases show the intended resolution split: coarse-code proximity can reveal a subject or isolated lexical cues, while exact answer selection depends on the clean representation evaluated inside the protected scoring stage.

\subsection{Attack Implementation Details}
\label{app:attacks}

We specify the attacker's observations, optimization procedure, and output for each evaluated attack.

\textbf{Search-Based Inversion.} We evaluate the embedding-guided adversarial-decoding and iterative-refinement stages of ZSInvert~\cite{zhang2025universal}. Given a released target representation, the attacker initializes Llama-3.1-8B-Instruct~\cite{llama2024herd} with the prompt ``Write a factual passage.'' At each decoding step, the LLM proposes its ten most likely next tokens for every active partial passage; the target encoder maps each expansion into the released representation space, similarity to the target ranks the expansions, and the best 50 remain in the beam. The highest-scoring completed passage becomes the seed for the next round, whose prompt asks for a factual passage similar to that seed. We run six search rounds with at most 80 token steps per round. Float targets use cosine similarity, while hash targets use normalized hash dot product, an affine transformation of Hamming distance. A separate MiniLM verifier~\cite{wang2020minilm} measures semantic similarity between the final reconstruction and the private target text; we report the mean over 100 sampled targets.

\textbf{Generation-Based Inversion.} Following GEIA~\cite{li2023geia}, the attacker collects auxiliary PersonaChat~\cite{zhang2018personachat} text--representation pairs and trains an embedding-conditioned DialoGPT-medium~\cite{zhang2020dialogpt} decoder. A learned linear layer maps each released representation to the decoder hidden dimension and prepends it as a pseudo-token before the ground-truth token embeddings; teacher-forced cross-entropy then trains both the projection and causal LM to predict the original passage. We train a separate decoder for each released representation for 10 epochs at a learning rate of $1\!\times\!10^{-5}$ and a batch size of 16. At test time, the held-out representation alone conditions width-5 beam decoding for at most 50 new tokens, producing one reconstruction per target; Table~\ref{tab:inversion-nodp} reports the test-split mean.

\textbf{Property Inference.} Following Song and Raghunathan~\cite{song2020information}, the attacker receives an auxiliary collection of Stage~1 releases with known attributes, fits supervised probes on those representation--attribute pairs, and predicts the attribute of independently held-out victim releases. We test four-way topic classification on AG News~\cite{zhang2015character}, binary sentiment on IMDB~\cite{maas2011learning}, and 50-way closed-set authorship attribution on 20 Newsgroups~\cite{lang1995newsweeder}. For authorship, the label is the email address in the \texttt{From:} header among the 50 most prolific authors; we remove the complete header block before encoding the message body. Every sample is encoded as the deployed Stage~1 query release using the current E5 hash checkpoint, the ``query:'' prefix, a 48-token limit, and the mechanism under evaluation. Each of five seeds creates stratified 60/20/20 attacker-training, validation, and victim-test partitions. Logistic regression, a two-hidden-layer MLP, and LightGBM~\cite{ke2017lightgbm} are fitted on the attacker-training split; validation macro F1 selects the architecture, and Table~\ref{tab:dp-pia} reports its victim-test macro F1 averaged across seeds.

\subsection{Qualitative Inversion Cases}
\label{app:inversion-cases}

Table~\ref{tab:eia-cases-full} reports the complete decoded strings for the case in Table~\ref{tab:eia-case-study} and three additional targets. Across the four cases, the float release supports reconstruction of specific entities, facts, and relations. The learned hash usually retains a broad topic or lexical association while dropping identifying details. Gaussian randomization removes even that stable association, producing outputs unrelated to the target.

\begin{table*}[t]
\centering
\caption{Complete search-based inversion outputs for four MS MARCO targets. Case A is the example summarized in Table~\ref{tab:eia-case-study}; Cases B--D are additional examples. Cosine is measured by the independent verifier.}
\label{tab:eia-cases-full}
\scriptsize
\setlength{\tabcolsep}{4pt}
\begin{tabular}{@{}cl>{\raggedright\arraybackslash}p{0.75\textwidth}r@{}}
\toprule
Case & Release & Complete target or reconstruction & Cosine \\
\midrule
A & Target & Welcome to the U.S. High School Bowling Foundation. The U.S. High School Bowling Foundation (USHSBF) is an IRS 501 (c)(3) non-profit organization that actively promotes the growth of high school bowling throughout the country. & -- \\
& Float & US Bowling Foundation. The nonprofit organization. USABF or HS, is an (non profit.)501, 501(C). US bowling ( foundation) or the bowling ( foundation) (U.s) & .903 \\
& Learned hash & US Bank is affiliated by National Bank Union(UFBG,NFUB or more informlly refererd UBFGU) through UBFBU-UnitedStatebannkinngunitiationthough there name oficialally stands out too - unitesstatebunkfinfinitation. However alot offpeople recognizeit with thershorthandletternamr ofunsttobbankor & .464 \\
& Gaussian, $\varepsilon=8$ & The country Australia encompasses expansive territories spanning over seven major regions with diverse geology, resulting form varied processes across historical timescale having several distinct natural environments within different territoories - Northern Territories comprising land that shares the Timoresse Sea and extends from southern coast to Indonesian borders as well northern territory known as Western Artesia, along the Great Victoria river system in addition the southeastern coastline forms extensive shoreline along & $-.060$ \\
\addlinespace
B & Target & What does FGA stand for? What does FGA mean? This page is about the various possible meanings of the acronym, abbreviation, shorthand or slang term: FGA. & -- \\
& Float & FAA acronym definition abbreviation. What meaning? short name F.A. Association meaning? What stands: stands for. & .697 \\
& Learned hash & The Aviation Alphabet Standards for pilots. A commonly accepted and well-diffused standardization was put forward with these phrases - Uniform - Mike Victor Lindberqg (L), Si-gmah Si-bay, Fo-sure-Fo-sho Fos- ter (R). & .204 \\
& Gaussian, $\varepsilon=8$ & An Intergregaler Frost (also sometimes spelled Yuki-On) event periodically transpires across Arctic archipelagos located within northern Norway towards summations occurring after equinocial solar events. & .031 \\
\addlinespace
C & Target & One pound of fat is 3,500 calories. If you simply eat 500 calories less per day, then in seven days that adds up to a 3,500 calorie deficit and you'd have lost one pound of fat. Fitness is science, not magic. (7 days x 500 calorie deficit = 3,500) & -- \\
& Float & seven days calories deficit per pound fat approximately =500 calorie workouts three * day eating one-pound less * is 500 calorie one day * equals fat one day loss * ( 500* (1 pound * ( 7 days * ( & .879 \\
& Learned hash & The ideal calorie deficit for fat and muscular pounds weightloss is approximately three pounds or so within eight weeks at rate equivalent of nearly one-further half pound in weight each seven-days. & .732 \\
& Gaussian, $\varepsilon=8$ & Antarctica Ross Land Sea ice field exhibits groundbreaking discoveries showcasing glaucial layers with crucial data suggesting expansive glacier formations reaching up to fifty seven-and twenty-two hundred-thoussd-year intervals dating around eleven and ninety eight miiiillion-earth years into the distant prehistirical. & $-.004$ \\
\addlinespace
D & Target & Alkali metals react with nonmetals to form ionic compounds. In these types of reactions, the alkali metal gives up its outermost electron to a nonmetal that is greedy for electrons. Reactions like this that involve an element exchanging an electron with another is called an oxidation-reduction or redox reaction. & -- \\
& Float & Reactive alkalide metals reaction nonmetal oxido make other. An Ionic forms bonds elements ox and compounds oxidic and metal nonreducible react with. Alternatively, reactive acid nonmetals make other compounds. & .824 \\
& Learned hash & Throughout chemistry ions form and transform, resulting in diverse mechanisms of conductivity modification as atoms exhibit changes of state along the way. & .413 \\
& Gaussian, $\varepsilon=8$ & The selenicerids (specific member from order of braniophorous gastreoids including species classified under Thylakocysticida) exhibit unique characteristics across various subtaxons found residing throughout multiple ecosystems such as cold sea floors adjacent volcanic regions and along with their habitat ranging near geothermic hot vent zones, abysses of ocean and river-mouth habitats located primarily close by island nations and & .052 \\
\bottomrule
\end{tabular}
\end{table*}

Cases A and B show how a low-bit release can preserve form-level cues such as an organizational acronym while losing the entity itself. Case C retains the broad calorie-deficit relation but changes the quantities and time scale, while Case D retains chemistry vocabulary but drops the defining electron-transfer and redox relation. In every case, the Gaussian reconstruction switches to a different subject, matching the near-zero verifier cosine.

\section{Additional Retrieval Results}
\label{app:dp-results}

Figure~\ref{fig:dp-k-sweep} shows how the candidate budget recovers ranking quality at the representative $\varepsilon=32$ operating point. Gaussian approaches the no-DP curve by $K=1000$ on all three datasets, while pure-vMF continues to benefit from larger candidate pools on NQ and FEVER. The horizontal float references expose both the remaining candidate loss and the point at which increasing $K$ saturates.

\begin{figure*}[thbp]
\centering
\includegraphics[width=\textwidth]{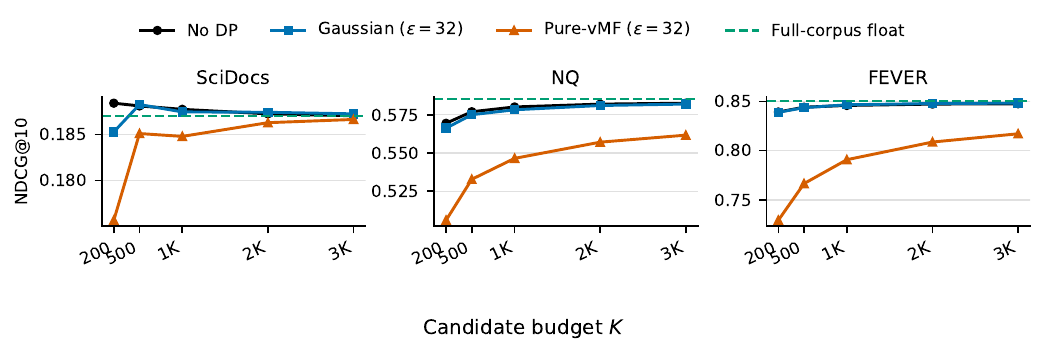}
\caption{E5 two-stage retrieval versus candidate budget. Gaussian and formal pure-vMF use $\varepsilon=32$.}
\label{fig:dp-k-sweep}
\end{figure*}

\begin{table}[t]
\centering
\caption{End-to-end RAG with client-side cross-encoder reranking on 500 NQ queries. The opened top ten are reranked to the five passages supplied to Qwen.}
\label{tab:rag-cross-encoder}
\footnotesize
\setlength{\tabcolsep}{3pt}
\begin{tabular}{@{}lcrrr@{}}
\toprule
Retrieval & $(\varepsilon,K)$ & Hit & EM & F1 \\
\midrule
Full-corpus float & $(\infty,N)$ & 89.2 & 53.2 & 65.36 \\
No-DP hash & $(\infty,500)$ & 88.6 & 53.6 & 65.68 \\
Gaussian & $(16,3000)$ & 89.0 & 53.6 & 65.71 \\
RDP-vMF & $(16,3000)$ & 89.2 & 53.2 & 65.35 \\
Pure-vMF & $(64,3000)$ & 89.0 & 53.2 & 65.34 \\
\bottomrule
\end{tabular}
\end{table}

\DPEFiveFullTable

\DPBGEFullTable

Tables~\ref{tab:dp-e5-full} and~\ref{tab:dp-bge-full} report every mechanism, privacy budget, dataset, and candidate budget used in the retrieval evaluation. Both encoders exhibit the same operating pattern: approximate mechanisms saturate at smaller $K$, while formal pure-vMF requires a larger pool at tighter budgets and converges toward the float reference as $\varepsilon$ increases.

\textbf{Tighter budgets with larger shortlists.} Table~\ref{tab:eps-k-frontier} extends the E5 sweep to $K=50{,}000$ on FEVER and NQ, the two largest corpora of the DP sweep. The protocol is unchanged; only $K$ grows. Pure-vMF at $\varepsilon=32$ and both approximate mechanisms at $\varepsilon=8$ retain 94--98\% of full-corpus NDCG@10 by $K=10{,}000$, which is 0.18\% of FEVER. Latency and traffic follow the measured grid in Appendix~\ref{app:latency-methodology}.

\begin{table}[H]
\centering
\caption{Privacy budget versus candidate budget: E5 Stage~2 NDCG@10 as a percentage of full-corpus float (FEVER / NQ), with measured traffic and latency at 128 Owner threads.}
\label{tab:eps-k-frontier}
\footnotesize
\setlength{\tabcolsep}{2.6pt}
\begin{tabular}{@{}lrrrrr@{}}
\toprule
Mechanism, $\varepsilon$ & $K$=3K & 5K & 10K & 20K$^{\ast}$ & 50K$^{\ast}$ \\
\midrule
Gaussian, 8 & 89/92 & 91/94 & 94/96 & 97/97 & 98/99 \\
RDP-vMF, 8 & 93/94 & 94/96 & 96/97 & 98/98 & 99/99 \\
Pure-vMF, 32 & 96/96 & 97/97 & 98/98 & 99/99 & 99/99 \\
\midrule
Traffic (MB) & 13.6 & 22.2 & 44.1 & 87.3 & 217.6 \\
1~Gbps (s) & 0.57 & 0.89 & 1.70 & 3.29 & 7.86 \\
100~Mbps (s) & 1.21 & 1.95 & 3.91 & 7.70 & 19.1 \\
\bottomrule
\end{tabular}
\begin{minipage}{\columnwidth}
\footnotesize Gaussian and RDP-vMF use $\delta=10^{-6}$. $^{\ast}$Uses the 256-MiB frame cap (Appendix~\ref{app:protocol-implementation}).
\end{minipage}
\end{table}

\subsection{Cross-Encoder Compatibility}
\label{app:rag}

The pretrained E5 scorer produces the top ten Natural Questions passages, matching the protocol's $k=10$ payload allowance. The User opens those passages through OT, locally reranks them with the MS MARCO MiniLM-L-6-v2 cross-encoder~\cite{wang2020minilm}, and supplies the top five to Qwen3-32B. The cross-encoder consumes the clean query and authorized plaintext payloads entirely on the User side.

Table~\ref{tab:rag-cross-encoder} shows that client-side reranking raises answer-bearing context coverage by 0.6–1.4 points and improves EM by at least 3.0 points. Every protected variant remains within 0.4 EM and 0.35 F1 of the cross-encoder float reference, so the learned filter composes with a standard retrieve–rerank–generate stack while retaining the protocol's payload-access bound.

\section{Implementation and Latency Measurement Details}
\label{app:tee_latency}

This appendix presents the complete message sequence (\S\ref{app:protocol-sequence}), protocol implementation (\S\ref{app:protocol-implementation}), latency methodology (\S\ref{app:latency-methodology}), and matched comparison with RemoteRAG~\cite{chang2025remoterag}, PANTHER~\cite{li2025panther}, and P$^2$RAG~\cite{ming2026p2rag} (\S\ref{app:protocol-comparison})~\cite{chang2025remoterag,li2025panther,ming2026p2rag}.

\subsection{Complete Message Sequence}
\label{app:protocol-sequence}

Figure~\ref{fig:protocol-sequence} expands the compact flow in Figure~\ref{fig:protocol-overview} into the complete authenticated round, including session setup, round binding, quota reservation, OT extension, payload delivery, and commit.

\begin{figure*}[t]
\centering
\includegraphics[width=\textwidth]{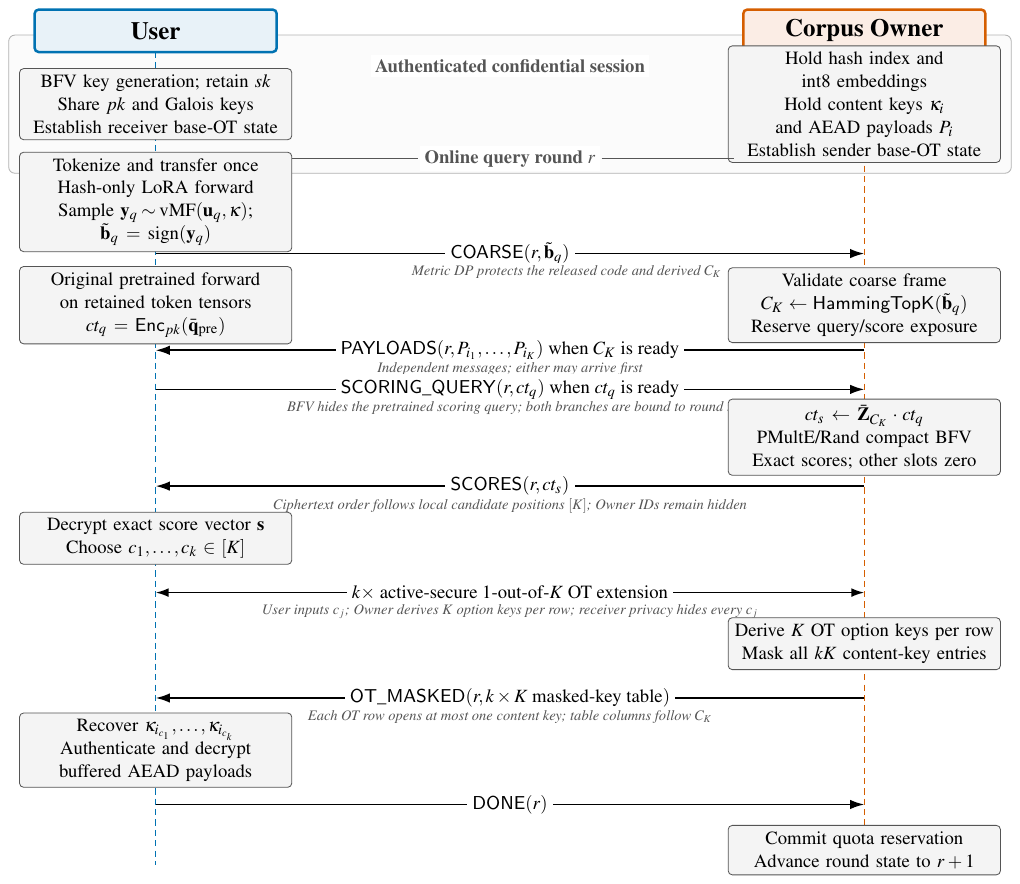}
\caption{Two-party message sequence for one authenticated query round. The coarse frame starts Owner-side Hamming search while the User runs the pretrained encoder and BFV encryption. Payload streaming begins when $C_K$ is ready, the scoring query is sent when $ct_q$ is ready, and either message may arrive first; active-secure OT later hides the selected positions and releases only their content keys.}
\label{fig:protocol-sequence}
\end{figure*}

\subsection{Protocol Implementation}
\label{app:protocol-implementation}

\paragraph{BFV scoring.} Following the Homomorphic Encryption Standard's 128-bit classical-security parameters~\cite{homomorphicEncryptionStandard2018}, the Microsoft SEAL implementation uses polynomial-modulus degree 8192, a 25-bit batching plaintext modulus, and a 180-bit coefficient-modulus chain, within the recommended 218-bit ceiling. The exact signed-int8 contract clamps both operands to $[-127,127]$, giving $|\langle\bar{\mathbf{q}},\bar{\mathbf{z}}\rangle|\le 127^2d$ and exact agreement with an int32 oracle for $d=768$. The compact layout uses coefficient-modulus bits $(50,40,40,50)$ and 1024-slot segments, masks each segment anchor, and rotates anchors into dense residues, so that one output ciphertext carries up to 8192 scores at the cost of a second multiplicative level. The Galois keys contain only the required rotation steps. Randomized evaluation of the compact layout follows Hwang et al.'s PMultE and Rand construction~\cite{hwang2025ciphertext}: it samples each plaintext lift from the appropriate modulo-$t$ discrete-Gaussian coset, adds the prescribed rounded continuous-Gaussian multiplication error, and adds one public-key randomizer before the public rotation, mask, and addition subcircuit. After the subcircuit, each output ciphertext receives a final randomizer whose \(c_0\) error has parameter \(s_f=4.6\,nt/2\approx2^{39.2}\), above \((nt/2)\,\eta_{2^{-80}}(\mathbb Z^{8192})\); it leaves the 23-bit remaining noise budget unchanged, and timing the builds with and without it back to back on the unshaped link changes median latency by $-3\%$ to $+1.5\%$ at $K=500$ and 3000, within run-to-run variation, so the reported timings stand. The discrete sampler uses a high-precision 192-bit cumulative distribution table (CDT) truncated at eight times its Gaussian parameter; tables are constructed once when the Owner daemon starts and reused across rounds. The User creates a fresh symmetric query ciphertext, while the Owner holds only public and rotation keys. Compact masking makes every non-score plaintext slot zero, and randomized evaluation makes the full ciphertext view simulatable as stated in Theorem~\ref{thm:he-user-view}. Randomized evaluation matters because the User holds the secret key: without it, the noise in the returned ciphertext depends on the Owner's plaintext operands, and the User could analyze that noise to learn about document embeddings beyond the scores. The final randomizer is needed because the mask is an integer-polynomial multiplication, which a secret-key User can invert: without it, the output noise of a single-group ciphertext reveals every pre-mask slot modulo \(t\), and hence the coordinate differences between neighbouring candidates. All reported latencies use this randomized evaluation.

\paragraph{Long-lived HE roles.} A User daemon retains the BFV context, public key, and secret key and handles symmetric query encryption and parallel score decryption; an Owner daemon retains the public context, public key, Galois keys, and randomized-evaluation samplers and handles candidate scoring. Both daemons load their key material once, bind each command to an exact round identifier, and terminate on malformed input or computation failure. The harness assigns 128 OpenMP threads to Owner scoring and 8 to User decryption, with one SEAL evaluator, decryptor, or batch encoder per worker where required.

\paragraph{Active-secure key transfer and payload protection.} The OT backend uses libOTe's OOS active-secure 1-out-of-$N$ extension~\cite{orru2016actively} with a 40-bit statistical check and a 16-bit option index, supporting project candidate budgets up to $K=60000$. Following the IKNP OT-extension organization~\cite{ishai2003extending}, the parties establish a small public-key base-OT correlation once per long-lived connection and derive each query's $k$ fresh extension rows with symmetric-key work. Each query advances the internal PRG state and performs a nonzero-challenge malicious check. The C++ sender derives and masks the full $k\times K$ option table before serialization, so raw option keys remain inside the Owner process. Each table entry XOR-masks a 128-bit content key with the first 128 bits of SHA-256 applied to a domain-separated OT key. Document plaintexts contain an encrypted true-length field, are padded in 4096-byte units, and are protected by one-shot ChaCha20--Poly1305 under an HKDF--SHA-256-derived key. RAG corpora are chunked into retrieval units of this size; for documents that cannot be chunked, PIR over $C_K$ would reduce payload traffic to $O(k)$ at extra Owner computation.

\paragraph{Network path and state.} The network design exchanges versioned frames carrying a session identifier, message type, round identifier, HE layout, $K$, and a length-capped payload. Session establishment binds the layout and $K$. A \textsc{Coarse} frame starts Hamming search while the User computes the pretrained representation; after $C_K$ is fixed, the Owner sends \textsc{Payloads} on the ordinary channel while the User sends the separate \textsc{Scoring-Query} frame and the Owner computes \textsc{Scores}. The parties then execute OT, the Owner sends \textsc{OT-Masked}, and the User opens the selected buffered payloads before \textsc{Done}. TCP delayed-ACK batching is disabled for these latency-sensitive control frames. Header fields and the 256-MiB payload cap are validated before payload allocation or blocking reads, and either-side failure closes the session and poisons reusable cryptographic state. With one 4096-byte payload block per candidate, the payload frame supports $K\le65{,}027$, so the OT bound $K\le60{,}000$ binds; the cap was 64~MiB ($K\le16{,}256$) before the $K=20{,}000$ and $50{,}000$ measurements, and frames for smaller $K$ are byte-identical.

\subsection{Latency Measurement Methodology}
\label{app:latency-methodology}

We use two complementary measurement experiments. The corpus-level comparison runs all four protocols against the same E5-base-v2 query and document embeddings on SciDocs, Webis-Touch\'e, and NQ-1M. A separate two-process harness exercises our complete Owner/User message path after query-embedding generation and checks every recovered payload byte-for-byte against the retrieval oracle. Following private-retrieval benchmarking practice, we separate reusable key/index preprocessing from steady-state online latency while including every per-query encryption, HE evaluation, OT extension, payload transfer, and serialized byte~\cite{li2025panther,ming2026p2rag}.

\paragraph{Latency and traffic accounting.} Let $T_{\mathrm{post}}$ contain encrypted-score return, local decryption and selection, OT, and masked-key delivery. The pipelined critical path is $T_{\mathrm{hash}}+T_{\mathrm{DP}}+\max\{T_{\mathrm{Hamming}}+T_{\mathrm{payload}},\max(T_{\mathrm{Hamming}},T_{\mathrm{pre}}+T_{\mathrm{enc}})+T_{\mathrm{score}}+T_{\mathrm{post}}\}+T_{\mathrm{open}}$, where transfer terms are charged at the evaluated bandwidth. Payload bytes therefore overlap scoring rather than appearing as a serial suffix. Figure~\ref{fig:latency} reports this path, while Figure~\ref{fig:rag-latency} adds Qwen3-32B generation. Query-side Gaussian, RDP-vMF, and pure-vMF release leave the subsequent message flow unchanged; the DP rows use the largest measured release time, 0.92 ms for vMF.

\paragraph{Measurement scope.} Long-lived User and Owner BFV daemons keep secret and public evaluation material in their respective roles, while the active-secure libOTe daemons reuse base-OT state across rounds. The protocol-stage experiment uses $N=25{,}657$ synthetic rows and checks the complete HE-score, OT-key, and AEAD-payload chain against an int32 and byte-for-byte oracle. The RAG experiment measures 20 full-corpus float queries, 30 exact Hamming queries over all 2.68M NQ codes, and generation for 100 NQ queries (348--1{,}406 input tokens, median 740), served one request at a time by vLLM 0.32 in BF16 with prefix caching disabled. Figure~\ref{fig:latency} charges exact serialized byte counts at 10~Gbps; Table~\ref{tab:e2e-k} and Table~\ref{tab:netem-grid} are measured end to end over shaped links.

\paragraph{Shaped-link measurements.} The User and Owner run as separate processes in a private network namespace whose loopback link tc/netem either leaves unshaped or limits to 1~Gbps or 100~Mbps with a 1-ms round-trip time. An Owner-thread scan at $K=10{,}000$ on the unshaped link gives 5.38, 3.05, 2.36, 2.03, 1.86, 1.72, and 2.21~seconds at 16, 32, 48, 64, 96, 128, and 192 threads; we use 128, the fastest, because 192 adds SMT siblings on a shared host. Table~\ref{tab:netem-grid} reports the full grid. A 1-Gbps link adds no measurable latency: at $K=3000$, shaping raises the Owner's payload sending from 7 to 82~ms, which overlaps 460--528~ms of HE scoring. At 100~Mbps, every measured value lies below the serial projection that adds the transfer time to the unshaped latency (e.g., 1.21 versus 1.67~seconds at $K=3000$). Small-$K$ 100-Mbps cells vary by up to 20\% between runs; the variation lies entirely in network receive time, not in HE or OT computation.

\begin{table}[H]
\centering
\caption{Measured protocol latency (ms, median of 20 rounds) over shaped links with 128 Owner threads in the deployed randomized-evaluation mode. Traffic is per query.}
\label{tab:netem-grid}
\footnotesize
\setlength{\tabcolsep}{3pt}
\begin{tabular}{@{}rrrrr@{}}
\toprule
$K$ & Traffic (MB) & Unshaped & 1~Gbps & 100~Mbps \\
\midrule
104 & 1.23 & 80 & 88 & 184 \\
292 & 2.04 & 92 & 108 & 263 \\
500 & 2.93 & 120 & 129 & 297 \\
1000 & 5.07 & 209 & 206 & 569 \\
1956 & 9.17 & 378 & 345 & 814 \\
3000 & 13.65 & 581 & 565 & 1207 \\
5000 & 22.2 & 931 & 892 & 1948 \\
10000 & 44.1 & 1937 & 1703 & 3911 \\
20000$^{\ast}$ & 87.3 & 3620 & 3286 & 7700 \\
50000$^{\ast}$ & 217.6 & 8664 & 7857 & 19137 \\
\bottomrule
\end{tabular}
\begin{minipage}{\columnwidth}
\footnotesize $^{\ast}$Uses the 256-MiB frame cap (Appendix~\ref{app:protocol-implementation}).
\end{minipage}
\end{table}

\paragraph{Hardware.} Measurements were run on an AMD EPYC 9654 server with NVIDIA A100 GPUs under Linux 6.14. Encoder timing uses the default BF16, merged LoRA weights, 10 warm-up queries, and 50 measured queries. The BFV Owner uses 128 OpenMP threads and the User decryptor uses 8; each network cell discards one warm-up round before 20 measured rounds and reports the median.

\begin{table}[H]
\centering
\caption{Qwen3-32B (vLLM) generation length and the $K=500$ protection overhead; seconds per query (median of 100).}
\label{tab:rag-latency-length}
\small
\begin{tabular}{@{}rrrrr@{}}
\toprule
Tokens & Generation & Plaintext & Protected & Overhead \\
\midrule
1 & .278 & .287 & .408 & 42.1\% \\
32 & 1.617 & 1.626 & 1.747 & 7.4\% \\
128 & 5.744 & 5.753 & 5.874 & 2.1\% \\
256 & 11.241 & 11.250 & 11.371 & 1.1\% \\
\bottomrule
\end{tabular}
\end{table}

\subsection{Matched Protocol Comparison}
\label{app:protocol-comparison}

\paragraph{Common benchmark contract.} All corpora use the same normalized 768-dimensional E5-base-v2~\cite{wang2022e5} vectors and top-10 output: 1,000 SciDocs queries over 25,657 documents, 49 Webis-Touch\'e queries over 382,545 documents, and all 3,452 NQ queries over NQ-1M, which retains every judged-relevant NQ document and fills the remaining positions from the corpus's original order.
Query encoding is excluded because it is identical across methods; online cryptographic setup, candidate search, secure scoring, selection, and serialized traffic are included. Long-lived keys and sessions exclude one-time setup. Figure~\ref{fig:protocol-comparison} demonstrates the exact traffic at 100~Mbps, while Table~\ref{tab:protocol-comparison-detail} uses 1~Gbps; the P$^2$RAG projection additionally charges 0.1~ms RTT for each declared online round because it requires two non-colluding servers.

\begin{table}[H]
\centering
\caption{Matched online latency at 1~Gbps in seconds per query. Ours uses the smallest $K$ retaining at least 99\% of float NDCG@10: 292, 104, and 1956, respectively. P$^2$RAG includes both client--server and inter-server transfer.}
\label{tab:protocol-comparison-detail}
\small
\begin{tabular}{@{}lrrr@{}}
\toprule
Method & SciDocs & Touch\'e & NQ-1M \\
\midrule
Ours & .109 & .090 & .349 \\
P$^2$RAG & .155 & 2.254 & 5.012 \\
RemoteRAG & 3.001 & 5.256 & 7.167 \\
PANTHER & 8.126 & 20.389 & OOM \\
Pisces & 16.6 & $\approx$175 & T/O \\
\bottomrule
\end{tabular}
\end{table}

\paragraph{Our protocol.} We use pure-vMF at $\varepsilon=64$, randomized-evaluation compact BFV scoring, and active-secure 10-out-of-$K$ key transfer. After unit-resolution refinement at each 99\% boundary, the smallest evaluated budgets are $K=292$ on SciDocs (0.18533 versus 0.18702), $K=104$ on Touch\'e (0.24753 versus 0.24954), and $K=1956$ on NQ-1M (0.62768 versus 0.63385). Each value adds the measured 32-thread full-index Hamming scan (0.20, 0.66, and 2.67~ms) and the 0.92-ms query DP to the tc/netem-measured protocol path at that $K$ (Table~\ref{tab:netem-grid}), in the deployed randomized-evaluation mode with 128 Owner threads.

\paragraph{P$^2$RAG~\cite{ming2026p2rag}.} We use the authors' official implementation\footnote{https://github.com/myl7/p2rag}, compile the authors' retrieval kernel and expose only $N$, $d$, and the bisection depth as runtime parameters. We set $k'=16$ so its revealed set contains the requested top-10, measure five online compute runs, and combine their median with the exact user--server and inter-server byte formulas from the paper's Table~2. Its kernel uses all 192 hardware threads, compared with 128 threads for our BFV scorer, and the trusted-dealer preprocessing remains offline as specified by P$^2$RAG; these two factors favor its online result.

\paragraph{RemoteRAG~\cite{chang2025remoterag}.} We reproduce the paper's spherical-cap shortlist geometry and Paillier encrypted cosine with a 1024-bit modulus, \texttt{gmpy2} modular exponentiation, and process-parallel query encryption and inner products. At $\varepsilon=15360$ and $k=10$, the mean shortlist sizes are 627.3 on SciDocs, 1,581.4 on Touch\'e, and 2,178 on the five measured NQ-1M queries; the first two corpora retain 100\% of the float top-10 in the shortlist. We report the median of three long-lived-key queries on SciDocs and Touch\'e and five on NQ-1M.

\paragraph{PANTHER~\cite{li2025panther}.} We reuse the authors' artifact\footnote{https://github.com/AntCPLab/OpenPanther}. We compile the authors' random-client/random-server secure path, and quantize the common E5 vectors with a corpus-calibrated 9-bit affine map. SciDocs uses 8,074 bounded clusters, a 1,561-item stash, a maximum cluster size of 10, and 1,280 probes. Touch\'e uses four bounded-cluster levels with 28,198, 10,280, 3,205, and 1,016 clusters, a 14,432-item stash, a maximum cluster size of 20, and 512/256/128/64 probes. These settings obtain 99.22\% and 99.18\% float top-10 agreement, respectively. NQ-1M uses 49,720 and 63,928 bounded clusters, a 37,125-item stash, a maximum cluster size of 20, and 792/396 probes, attaining 99.0\% float top-10 agreement over 100 queries. The 768-dimensional, 20-point cluster layout expands its 113,648 PIR records to 15,480 32-bit elements each; the implementation maps 1,188 probes to 1,782 cuckoo bins and generates their SEAL replies in parallel while retaining the encoded database. This answer-stage memory peak exhausts the 256-GB host after database construction and the distance, argmin, and garbled-circuit stages complete, so we report OOM rather than an extrapolated latency. The completed SciDocs and Touch\'e values are medians of three secure runs with measured traffic.

\paragraph{Pisces~\cite{liang2026pisces}.} We build the authors' artifact\footnote{https://github.com/liang-xiaojian/Pisces} (upstream commit a581a5e) and run its dense-only similarity path with the default parameters, $K=10$, and 20 queries, with both parties on one host. Its fuzzy-PSI filter keeps 91\% of SciDocs and 85.5\% of Touch\'e documents for MPC scoring because E5 embeddings are anisotropic, so the online cost still grows with $N$: a query takes 10.7 and 89.8~seconds of computation and transfers 706 and 10{,}158~MiB on the two corpora, with top-10 overlap .98 against plaintext E5. Mean-centering the embeddings shrinks the SciDocs candidate set to about 12\%, but top-10 overlap then falls to .48, and the PIR-to-share stage still costs 3.8--5.9~seconds. On NQ-1M, the first query did not finish within 27 minutes; we report it as a timeout.

\end{document}